\documentclass[10.6pt, reqno]{article}
\usepackage[utf8]{inputenc}
\usepackage{amsmath, amsthm}
\usepackage{amsfonts}
\usepackage{mathrsfs}
\usepackage{amssymb}
\usepackage{soul}
\usepackage{bbm}
\usepackage{todonotes}
\usepackage{soul}
\usepackage{authblk}
\usepackage{bm}
\usepackage{color}
\usepackage{float}
\usepackage{authblk}
\usepackage{enumerate}

\usepackage{enumerate}
\usepackage{tikz}
\usepackage{amsmath,amssymb}
\usetikzlibrary{calc}

\usepackage{caption}
\usepackage{hyperref}
\usepackage{slashed}
\hypersetup{hypertex=true,
colorlinks = black,
linkcolor =black,
anchorcolor = black,
citecolor = black}
\usepackage{graphicx}
\usepackage{cancel}
\usepackage{mathtools}
\usepackage{subfigure}
\usepackage{slashed}
\usepackage{amsfonts}
\usepackage{amssymb}
\usepackage{accents}
\usepackage{amsmath}
\usepackage{amsxtra}
\usepackage{epstopdf}
\usepackage{latexsym}
\usepackage{mathrsfs}
\usepackage{leftidx}
\usepackage{tensor}
\usepackage{enumitem}
\newtheorem{theorem}{Theorem}[section]
\newtheorem{lemma}[theorem]{Lemma}

\newtheorem{proposition}[theorem]{Proposition}
\newtheorem{conjecture}[theorem]{Conjecture}
\newtheorem{definition}[theorem]{Definition}

\theoremstyle{remark}
\newtheorem{remark}[theorem]{Remark}

\numberwithin{equation}{section}
\begin{document}
\title{Exponentially-growing Mode Instability on the Reissner--Nordström-Anti-de-Sitter Black Holes}

\author{Weihao Zheng\thanks{wz344@math.rutgers.edu}}
\affil{\small Department of Mathematics, Rutgers University, Hill Center, 110 Frelinghuysen Road, Piscataway, NJ, USA}
\date{\today}
\maketitle
\begin{abstract}
We construct exponentially growing mode solutions to the uncharged and charged Klein--Gordon equations on the $(3+1)$-dimensional sub-extremal Reissner--Nordström-anti-de-Sitter (AdS) spacetime under reflecting (Dirichlet or Neumann) boundary conditions. Our result applies to a range of Klein--Gordon masses above the so-called Breitenlohner--Freedman bound, notably including the conformal mass case. The mode instability of the Reissner--Nordström-AdS spacetime for some black hole parameters is in sharp contrast to the Schwarz-\ schild-AdS spacetime, where the solution to the Klein--Gordon equation is known to decay in time. Contrary to other mode instability results on the Kerr and Kerr-AdS spacetimes, our exponentially growing mode solutions of the uncharged and weakly charged Klein--Gordon equation exist independently of the occurrence or absence of superradiance. We discover a novel mechanism leading to an exponentially growing mode solution, namely, a near-extremal instability for the Klein--Gordon equation. Our result seems to be the first rigorous mathematical realization of this instability.
\end{abstract}

\section{Introduction}
In this paper, we are interested in constructing exponentially growing mode solutions to the Klein--Gordon equation with parameters $(M,e,\Lambda,\alpha,q_{0})$ on a sub-extremal Reissner--Nordström-anti-de-Sitter background:\begin{equation}
g_{RN}^{\mu\nu}D_{\mu}D_{\nu}\phi =\bigl(-\frac{\Lambda}{3}\bigr)\alpha\phi,\quad D_{\mu} = \nabla_{\mu}+iq_{0}A_{\mu},\label{Klein-Gordon}
\end{equation}
where $A$ is the electromagnetic potential and the Reissner--Nordström-AdS metric takes the form of\begin{equation}
\label{RN metric}
g_{RN} = -\Omega^{2}dt^{2}+\frac{1}{\Omega^{2}}dr^{2}+r^{2}d\sigma^{2},\quad\
\Omega^{2} = 1-\frac{2M}{r}+\frac{e^{2}}{r^{2}}+\bigl(-\frac{\Lambda}{3}\bigr)r^{2}.
\end{equation}
Here $M>0$ is the mass of the black hole, $e\in\mathbb{R}$ is the charge of the black hole, $\Lambda: = -\frac{3}{l^{2}}<0$ is the cosmological constant, $q_{0}\in\mathbb{R}$ is the charge of the scalar field, and $\alpha<0$ is the negative Klein--Gordon mass. We will consider the charged scalar field case $q_{0}\neq0$ and the uncharged scalar field case $q_{0} = 0$. We will assume the Reissner--Nordström-AdS spacetime is sub-extremal, i.e., $\Omega^{2}$ admits two positive roots. Let $r_{+} := r_{+}(M,e,l)$ be the largest root of $\Omega^{2}$ corresponding to the area radius of the event horizon.

We will make a gauge choice for the electromagnetic potential $A$ of the following form\begin{equation}
A = A_{t}dt: =  -e\left(\frac{1}{r_{+}}-\frac{1}{r}\right)dt.\label{equation for A}
\end{equation}
In the case $q_{0} = 0$, equation $\eqref{Klein-Gordon}$ is reduced to the well-known uncharged Klein--Gordon equation with the usual covariant derivative $D_{\mu} = \nabla_{\mu}$.

Due to the lack of global hyperbolicity of the asymptotically AdS spacetime, the natural formulation of the Klein--Gordon equation $\eqref{Klein-Gordon}$ is the initial-boundary value problem; see \cite{holzegel2012self,holzegel2012well}. We also refer to Section $\ref{geometry}$ for a detailed introduction to the geometry and boundary conditions of the asymptotically AdS spacetime. The growing mode solutions we consider in this paper take the form of $\phi(t,r) = e^{i\omega t}\psi(r)$, where $\omega\in\mathbb{C}$ has a negative imaginary part and $\phi$ is a function regular at the event horizon. 

In view of the fact that $\eqref{RN metric}$ can be parametrized by $(M,e,l)$ or $(M,r_{+},l)$ alternatively, we will call $(M,r_{+},l)$ the sub-extremal parameters of $\eqref{RN metric}$ if the corresponding $(M,e,l)$ are sub-extremal parameters. Then under these new parameters, $e = 0$ corresponds to $M = M_{e = 0}: = \frac{r_{+}}{2}\left(1+\frac{r_{+}^{2}}{l^{2}}\right)$ while the extremality corresponds to $M = M_{0}:=r_{+}+\frac{r_{+}^{3}}{l^{2}}$. For fixed $(r_{+},l)$, the admissible sub-extremal range of $M$ is $M_{e = 0}\leq M<M_{0}$. We refer to Section \ref{section:parameters} for a detailed discussion of the parameters transform.

In the following, we will call $(M,r_{+},l,\alpha,q_{0})$ parameters of the Klein--Gordon equation $\eqref{Klein-Gordon}$. Our main results are as follows:

\begin{theorem}\textup{[Rough version of the main result]}
\label{rough}
For the Klein--Gordon equation $\eqref{Klein-Gordon}$, let $C_{DN} = 0$ for Dirichlet boundary conditions and $C_{DN} = -\frac{5}{4}$ for Neumann boundary conditions respectively. Imposing a reflecting boundary condition (Dirichlet or Neumann) for \eqref{Klein-Gordon}, we have the following three results about exponentially growing mode solutions:
\begin{enumerate}
  \item[(1)]\textup{(Large charge case)} Assume $(M_{b},r_{+},l)$ are given sub-extremal parameters with $M_{e = 0}<M_{b}<M_{0}$, $-\frac{9}{4}<\alpha<C_{DN}$ is fixed, and $\vert q_{0}\vert$ is a fixed large charge (the size of $\vert q_{0}\vert$ depends on $(M_{b},r_{+},l)$ and $\alpha$). Let $\mathcal{S}(M_{b},r_{+},l,\alpha,q_{0})$ be the set of all $M_{e = 0}\leq M<M_{b}$ such that $(M,r_{+},l)$ are admissible sub-extremal parameters and $\eqref{Klein-Gordon}$ with parameters $(M,r_{+},l,\alpha,q_{0})$ has an exponentially growing mode solution that is regular at the event horizon. Then $\mathcal{S}(M_{b},r_{+},l,\alpha,q_{0})$ is non-empty and open.
  \item[(2)]\textup{(General fixed charge case)} For any fixed $(r_{+},l,\alpha,q_{0})$ satisfying the conditions
  \begin{align}
  &-\frac{9}{4}<\alpha<\min\left\{C_{DN},-\frac{3}{2}+\frac{1}{2}l^{2}q_{0}^{2}\right\},\label{large1}\\
  &\frac{r_{+}^{2}}{l^{2}}>R_{0}\left(\alpha,l^{2}q_{0}^{2}\right),\label{large}
  \end{align}
  where $R_{0}$ is the positive solution to the quadratic equation\begin{equation}
  24\left(\alpha+\frac{3}{2}-\frac{1}{2}l^{2}q_{0}^{2}\right)x^{2}+\left(4\left(\alpha+\frac{3}{2}-\frac{1}{2}l^{2}q_{0}^{2}\right)-2l^{2}q_{0}^{2}+6\right)x+1 = 0,\label{quadratic}
  \end{equation}
  let $\mathcal{S}_{0}(r_{+},l,\alpha,q_{0})$ be the set of all $M_{e = 0}\leq M<M_{0}$ such that $(M,r_{+},l)$ are admissile sub-extremal parameters and $\eqref{Klein-Gordon}$ with parameters $(M,r_{+},l,\alpha,q_{0})$ has an exponentially growing mode solution that is regular at the event horizon. Then $\mathcal{S}_{0}(r_{+},l,\alpha,q_{0})$ is non-empty and open. 
  \item[(3)]\textup{(Weakly charged case)} For any given $(r_{+},l,\alpha)$ satisfies \eqref{large1} and \eqref{large} with $q_{0} = 0$, there exists $\epsilon>0$ small enough, such that for all $M\in(M_{0}-\epsilon,M_{0})$, there exists $\delta>0$ depending on $M$, such that $M\in\mathcal{S}_{0}(r_{+},l,\alpha,q_{0})$ for all $\vert q_{0}\vert\leq\delta$. 
\end{enumerate}
\end{theorem}

We make the following remarks.

\begin{remark}
In the large charge case, $\vert q_{0}\vert$ above is a constant depending on the given sub-extremal parameters, and growing mode solutions are constructed for spacetimes away from the extremality (in the sense that $M_{e = 0}<M<M_{b}<M_{0}$). For the general fixed charge case, however, there is no lower-bound requirement for $q_{0}$ and thus our result includes the uncharged case $q_{0} = 0$. Our general strategy in proving the large charge case and general fixed charge case is to construct growing mode solutions around a stationary solution and one should think that $M\in\mathcal{S}_{0}(r_{+},l,\alpha,q_{0})$ is close to $M_{c}$ where a stationary solution to \eqref{Klein-Gordon} exists. For the weakly charged case, we can further construct growing mode solutions not only near a stationary solution but for all $M$ close to $M_{0}$. For the uncharged case, the conditions \eqref{large1} and \eqref{large} are reduced to \begin{align}
&-\frac{9}{4}<\alpha<-\frac{3}{2},\label{uncharged range}\\&
\frac{r_{+}^{2}}{l^{2}}>\frac{1}{4(-\frac{3}{2}-\alpha)},
\end{align}
which includes the conformal mass $\alpha = -2$.
\end{remark}
\begin{remark}
\label{rmk:large charge}
The construction of growing mode solutions for the uncharged case crucially relies on the non-positivity of the energy on spacelike hypersurfaces $\{t = const\}$. As will be discussed in detail in Section \ref{uncharged instability}, the condition\begin{equation}
\frac{e^{2}}{r_{+}^{2}}<1+\frac{r_{+}^{2}}{l^{2}},\label{analogous Hawking reall bound}
\end{equation}
ensures that the energy remains positive and bounded on each constant time slice. Consequently, the boundedness of the solution to the uncharged Klein--Gordon equation \eqref{Klein-Gordon} can be established using standard arguments. Therefore, for growing mode solutions in the uncharged case, the parameters must satisfy the following necessary ``bound-violating'' condition:
\begin{equation}
\frac{e^{2}}{r_{+}^{2}}>1+\frac{r_{+}^{2}}{l^{2}},\label{large charge condition initially}
\end{equation}
which can be interpreted as a largeness requirement for the black hole charge; see Section \ref{uncharged instability} for a detailed discussion.
\end{remark}
\begin{remark}
The issue mentioned in Remark \ref{rmk:large charge} is first identified in the work of Holzegel and Warnick \cite{holzegel2014boundedness}, where the boundedness of solutions to the Klein--Gordon equation on general asymptotically AdS spacetime without Maxwell under given boundary conditions is established, providing the coercivity of the initial energy. They further constructed super-polynomially growing solutions if the initial energy is non-positive. We emphasize here that the boundary condition is crucial to the coercivity of the initial energy, since the non-positive energy can always be achieved by one of the boundary conditions on general asymptotically AdS spacetimes. In contrast to the results in \cite{holzegel2014boundedness}, our results consider the effect of the Maxwell field and further show that the negative initial energy can be constructed for the Reissner--Nordstrom-AdS spacetime under reflecting boundary conditions, and there exist solutions growing exponentially, as opposed to merely super-polynomially.
\end{remark}
\begin{remark}
As mentioned above, the growing mode solutions constructed in this paper for the large charge case and the general fixed charge case are obtained by perturbing stationary solutions. It remains unclear to us whether growing mode solutions exist for the general fixed charge and the black hole being close to extremality. However, by the computation in Section \ref{uncharged instability}, we can show that if\begin{equation}
\left(1+4l^{2}q_{0}^{2}\right)\frac{e^{2}}{r_{+}^{2}}<1+\frac{r_{+}^{2}}{l^{2}},\label{strong coupling}
\end{equation}
then there is no stationary solution to \eqref{Klein-Gordon}. Hence the necessary condition for the existence of stationary solutions is \begin{equation}
\left(1+4l^{2}q_{0}^{2}\right)\frac{e^{2}}{r_{+}^{2}}>1+\frac{r_{+}^{2}}{l^{2}}.
\end{equation}
\end{remark}

\paragraph{The role of the black hole charge $e$ in the existence of growing mode solutions}
It is important to note that, in direct contrast to the condition \eqref{large charge condition initially}, the decay results in \cite{holzegel2013stability,holzegel2013decay} for Klein--Gordon equations on Schwarzschild-AdS and Kerr-AdS imply that for any parameters $(r_{+},l,\alpha,q_{0})$\begin{equation*}
M_{e = 0}\notin\mathcal{S}_{0}(r_{+},l,\alpha,q_{0}),\quad M_{e = 0}\notin \mathcal{S}(M_{b},r_{+},l,\alpha,q_{0}).
\end{equation*}
In other words, it is impossible to construct the growing mode solutions if the black hole spacetime itself is uncharged ($e = 0$) and spherically symmetric; see Section \ref{sec: linear instability} for a detailed discussion. 

\paragraph{The role of the scalar field charge $q_{0}$ in the existence of growing mode solutions}
We now distinguish between the charged and uncharged cases. For the charged case, since the charge $\vert q_{0}\vert$ can be taken to be arbitrarily large, we have\begin{equation*}
\inf_{q_{0}}\inf\mathcal{S}(M_{b},r_{+},l,\alpha,q_{0}) = M_{e = 0}.
\end{equation*}
The presence of a (large) scalar field charge allows the black hole to be away from extremality and the black hole spacetime to be weakly charged. While for the general fixed charge case, by $\eqref{large}$ and the method in this paper, we can show\begin{equation*}
\inf_{(r_{+},l^{2})}\sup_{M\in\mathcal{S}_{0}(r_{+},l,\alpha,q_{0})}\frac{d\Omega^{2}}{dr}\Bigl|_{r = r_{+}} = 0,
\end{equation*}
which can be interpreted as follows: for a fixed charge $q_{0}$, there exist parameters $$ \bigl(M\in\mathcal{S}_{0}(r_{+},l,\alpha,0),r_{+},l\bigr)$$ very close to the extremal parameters. The result in the weakly charged case implies that growing mode solutions exist for all $(M,r_{+},l)$ close to the extremality.

Heuristically, the existence of growing mode solutions for large charge case in Theorem \ref{rough} can be explained in view of the fact that the charge $q_{0}$ plays a role in the so-called ``effective mass'' $\alpha_{eff}(r)$, i.e. the coefficient of zero order term of $\phi$ in $\eqref{Klein-Gordon}$. The effective mass $\alpha_{eff}(r)$ violates the Breitenlohner--Freedman bound in the compact region of $(r_{+},\infty)$ when the charge $q_{0}$ is large while keeping $\lim_{r\rightarrow\infty}\alpha_{eff}(r)$ satisfying the Breitenlohner--Freedman bound; see \cite{gubser2005phase,gubser2008breaking,hartnoll2008building,horowitz2011introduction} for the heuristical discussions. For the general fixed charge case and weakly charged case, the situation is more subtle since no large effective mass is available in general.
\paragraph{The Breitenlohner--Freedman bound and local well-posedness}
The bound $\alpha>-\frac{9}{4}$ is called the Breitenlohner--Freedman bound, which is crucial for local well-posedness results of $\eqref{Klein-Gordon}$. The equation $\eqref{Klein-Gordon}$ is proved to be local well-posed under Dirichlet boundary conditions \cite{holzegel2010massive} for $-\frac{9}{4}<\alpha<0$ and under Neumann boundary conditions \cite{warnick2013massive} for $-\frac{9}{4}<\alpha<-\frac{5}{4}$; see also Section \ref{sec:local well-posedness}. The ranges of our $\alpha$ in all three cases in Theorem \ref{rough} (condition $\eqref{large1}$) fit within both local well-posedness regimes for the reflecting boundary conditions. Furthermore, we note $\eqref{uncharged range}$ is conjectured in \cite{horowitz2011introduction} as the optimal mass range for the uncharged Klein--Gordon equation under which the growing mode solutions exist; see also Section \ref{physics}. However, condition $\eqref{large}$ is new to the best of our knowledge.
\paragraph{A toy model for the linearized Einstein--Maxwell equations}
The negative Klein--Gordon mass $\alpha$ in the setting is motivated by considering the Einstein--Maxwell equations with $\Lambda<0$:\begin{equation*}
Ric_{\mu\nu}(g) = \Lambda g_{\mu\nu}+T_{\mu\nu}^{EM},
\end{equation*}
where $T_{\mu\nu}^{EM}$ is the energy tensor given by the Maxwell field. Informally, $Ric$ can be viewed as a wave operator $\Box$ under wave harmonic coordinates. Hence, a toy model for studying the (in)stability of the Reissner--Nordström-AdS spacetime is the uncharged Klein--Gordon equation with a negative mass $\alpha$. Furthermore, $\eqref{Klein-Gordon}$ with the conformal mass $\alpha = -2$ can be viewed as a simplification of the analogous Teukolsky equations derived in the Reissner--Nordström-AdS spacetime, which plays an important role in the study of linear stability of the Reissner--Nordström-AdS spacetime. One can refer to \cite{graf2024linear1, graf2024linear2,graf2024linear} for the sharp decay results for the analogous Teukolsky equations on the Schwarzschild-AdS spacetime and \cite{giorgi2020linear} for the linear stability on the asymptotically flat Reissner--Nordström spacetime.
\paragraph{Construction of hairy black holes}
In \cite{holzegel2013stability}, Holzegel and Smulevici proved that any spherically symmetric solutions to the Einstein--Klein--Gordon equations with a negative cosmological constant $\Lambda$ in the vicinity of Schwarzschild-AdS converge exponentially to the Schwarzschild-AdS again. In particular, the scalar field $\phi$ decays exponentially in time. Their result shows that the only stationary spherical black hole solutions of the Einstein--Klein--Gordon equations with a negative cosmological constant $\Lambda$ are Schwarzschild-AdS spacetimes with vanishing scalar field $\phi$. However, using the mode solution construction obtained in this paper, which essentially relies on the non-zero black hole charge $e$, see Remark \ref{rmk:large charge} and \ref{large charge condition initially}, we show the existence of stationary black hole with a non-trivial (charged or uncharged) scalar field $\phi$ under the reflecting boundary conditions for the Einstein--Maxwell--Klein--Gordon equations with a negative cosmological constant $\Lambda$ in our companion work \cite{weihaozheng}. 

\paragraph{Strongly charged and weakly charged instability mechanism}
In \cite{shlapentokh2014exponentially} and \cite{dold2017unstable}, growing mode solutions were constructed respectively to the Klein--Gordon equation on the Kerr and Kerr-AdS spacetimes violating the Hawking--Reall bound. The mechanism behind these instabilities is of a superradiant nature for spacetimes outside of spherical symmetric. In \cite{besset2021existence}, growing mode solutions to the charged Klein--Gordon equation on Kerr--Newman-dS and Reissner--Nordström-dS have also been constructed, due to the strong coupling of the black hole charge and the scalar field charge. The mechanism behind this instability is also of a superradiant nature induced by the scalar field charge \cite{bachelot2004superradiance}; see Section \ref{previous work on mode instability} for a detailed discussion. 

As mentioned above, the instability mechanism for our large charge case is due to the effective mass violating the Breitenlohner--Freedman bound. In view of \eqref{strong coupling}, the coupling of the black hole charge $e$ and the scalar field charge $q_{0}$ is also strong in our large charge case, which is very similar to the case of growing mode solutions on Kerr--Newman-dS. This instability is indeed of superradiant nature \cite{di2015superradiance} and is called tachyonic instability in physics literature \cite{gubser2005phase,gubser2008breaking,hartnoll2008building,horowitz2011introduction,brito2020superradiance}. However, since there is no superradiance for the uncharged scalar field in the Reissner--Nordström-AdS spacetime, our growing mode solution in this case is due to a new mechanism called near-extremal instability, as already discussed in \cite{horowitz2011introduction}; see also Section \ref{uncharged instability} for a further discussion of this mechanism.

\paragraph{Outline of the rest of the introduction}
In Section $\ref{previous result on instability}$, we discuss some previous results regarding (in)stability on the asymptotically AdS spacetime and establish a connection between the growing mode solutions we get in this paper and an instability conjecture of the Reissner--Nordström-AdS spacetime. In Section $\ref{previous work on mode instability}$, we review previous works on constructing growing mode solutions to the Klein--Gordon equation on different spacetimes. We also give a brief introduction to the instability mechanism of these growing mode solutions there. In Section \ref{uncharged instability}, we discuss the novel instability mechanism leading to the growing mode solutions of the uncharged Klein--Gordon equation on the Reissner--Nordström-AdS spacetime. In Section \ref{physics}, we discuss the physics motivation and heuristic argument.

\subsection{Previous results on stability/instability of the asymptotically AdS spacetime and decay of the field}
\label{previous result on instability}
In the past decades, despite intensive research aimed at proving the stability of the Einstein vacuum equation in the asymptotically flat spacetimes, there have been only a few results regarding the (in)stability of the asymptotically AdS spacetime. Due to the presence of a conformal timelike boundary, the choice of boundary conditions plays a decisive role in the (in)stability issues. To address the nonlinear (in)stability problems in the asymptotically AdS spacetime, the first step is to understand the decay and boundedness properties of the linearized field equation, which will heuristically suggest the (in)stability of the spacetime. As mentioned above, a natural toy model for the Einstein equations (coupled with a matter field) with a negative cosmological constant is the Klein--Gordon equation on a fixed asymptotically AdS spacetime with a negative conformal mass. 

First, we give a brief definition of boundary conditions for the Klein--Gordon equation with a conformal mass $\alpha = -2$ here; see Section $\ref{sec: boundary conditions}$ for the definition for more general Klein--Gordon masses.
\begin{definition}
For the Klein--Gordon equation $\eqref{Klein-Gordon}$ on an asymptotically AdS spacetime, the solution $\phi$ satisfies the Dirichlet, Neumann, or dissipative boundary condition if the following holds:
\begin{enumerate}
  \item[(1)] Dirichlet boundary condition:\begin{equation*}
  r\phi\rightarrow 0,\quad r\rightarrow \infty.
  \end{equation*}
  \item[(2)] Neumann boundary condition:\begin{equation*}
  r^{2}\frac{\partial}{\partial r}\left(r\phi\right) = 0,\quad r\rightarrow\infty.
  \end{equation*}
  \item[(3)] Optimally dissipative boundary condition:\begin{equation*}
  \frac{\partial(r\phi)}{\partial t}+r^{2}\frac{\partial (r\phi)}{\partial r}\rightarrow 0, \quad r\rightarrow \infty.
  \end{equation*}
\end{enumerate}
\end{definition}
We define the reflecting boundary condition to be either the Dirichlet boundary condition or the Neumann boundary condition. Note that it is also possible to consider linear combinations of the Dirichlet and the Neumann boundary conditions, which are called the Robin boundary conditions. We will not pursue results under Robin boundary conditions in this paper.
\subsubsection{Local well-posedness results}
\label{sec:local well-posedness}
The study of the Klein--Gordon equation on asymptotically AdS spacetimes has been initiated in a series of works \cite{holzegel2010massive, holzegel2012well}. In \cite{holzegel2010massive}, Holzegel proved the local well-posedness results for Klein--Gordon equations on general asymptotically AdS spacetimes under the Dirichlet boundary conditions with $-\frac{9}{4}<\alpha<0$. For the Neumann boundary conditions, the local well-posedness results hold \cite{warnick2013massive} for $-\frac{9}{4}<\alpha<-\frac{5}{4}$. The work \cite{gannot2018elliptic} establishes the local well-posedness under the dissipative boundary conditions for some range of $\alpha$. In \cite{holzegel2012well}, Holzegel and Smulevici proved the local well-posedness results for the spherically symmetric Einstein--Klein--Gordon equations with a negative cosmological constant $\Lambda$ under Dirichlet boundary conditions with $-2\leq \alpha<0$. 

\subsubsection{Results and conjectures on the stability of the pure AdS spacetime with dissipative boundary conditions}
The pure AdS spacetime takes the form of\begin{equation*}
g_{AdS} = -\left(1+\frac{r^{2}}{l^{2}}\right)dt^{2}+\left(1+\frac{r^{2}}{l^{2}}\right)^{-1}dr^{2}+r^{2}d\sigma^{2}.
\end{equation*}
In \cite{holzegel2020asymptotic}, the following conjecture was made for the pure AdS with optimally dissipative boundary conditions.
\begin{conjecture}[Holzegel--Luk--Smulevici--Warnick,\cite{holzegel2020asymptotic}]
\label{cj1s}
Anti-de Sitter spacetime is asymptotically stable for optimally dissipative boundary conditions.
\end{conjecture}
Holzegel--Luk--Smulevici--Warnick showed the standard energy boundedness and integrated decay result for the Klein--Gordon equation with $\alpha = -2$ under optimally dissipative boundary conditions on the pure AdS spacetime in \cite{holzegel2020asymptotic}, which is the first step toward the proof of Conjecture $\ref{cj1s}$.
\subsubsection{Results and conjectures on the instability of the pure AdS spacetime with reflecting boundary conditions}
In \cite{dafermos2006dynamic,dafermos2006nonlinear}, the following instability conjecture about the pure AdS spacetime with reflecting boundary conditions was made by Holzegel and Dafermos.
\begin{conjecture}[Dafermos--Holzegel, \cite{dafermos2006dynamic,dafermos2006nonlinear}]
Anti-de Sitter spacetime is non-linearly unstable for reflecting (Dirichlet or Neumann) boundary conditions.
\label{cj1}
\end{conjecture}
The above conjecture relies on the expectation that, for any small perturbations to pure AdS initial data, the solutions to the Einstein vacuum equations with a negative cosmological constant $\Lambda$ will form a trapped surface region. The first numerical evidence for the above conjecture was given in \cite{bizon2011weakly}, where the authors showed numerically that pure AdS spacetimes are unstable as solutions to Einstein-scalar field equations under arbitrarily small generic perturbations. In a series of works \cite{moschidis2018characteristic,moschidis2017einstein,moschidis2020proof,moschidis2023proof}, Moschidis proved the instability of AdS for the Einstein-null dust system and Einstein-massless Vlasov system with reflecting boundary conditions under the spherically symmetric assumption.

\subsubsection{Linear (in)stability of AdS black hole}
\label{sec: linear instability}
In \cite{holzegel2013decay, holzegel2014quasimodes}, Holzegel and Smulevici showed a sharp inverse logarithmic decay rate for the Klein--Gordon equation on the Schwarzschild-AdS spacetime \footnote{In fact, under the spherically symmetric assumption of $\phi$, they can further show $\phi$ decays exponentially in time.} under the Dirichlet boundary condition. Moreover, their sharp logarithmic decay result also holds for the Kerr-AdS spacetime satisfying the Hawking--Reall bound. One should note that the Hawing--Reall bound is crucial in obtaining the decay. In fact, Dold \cite{dold2017unstable} constructed growing mode solutions of the Klein--Gordon equation on the Kerr-AdS spacetime violating the Hawking--Reall bound.

Once the decay result for the Klein--Gordon equation is obtained, the next step is to study the (in)stability of linearized gravity. For the Kerr spacetime, Teukolsky derived two decoupled equations for the linearized curvature components in \cite{teukolsky1972rotating}. Graf and Holzegel proved mode stability for the analogous Teukolsky equations on the Kerr-AdS spacetime satisfying the Hawking--Reall bound in their work \cite{graf2023mode}, thus ruling out the possibility of growing mode solutions. Most recently, in a series of works \cite{graf2024linear1,graf2024linear2,graf2024linear} establishing the linear stability of Schwarzschild-AdS, Holzegel and Graf showed that under Dirichlet-type boundary conditions, the solutions to the Teukolsky equations with perturbed Schwarzschild-AdS initial data converge to Kerr-AdS with an inverse logarithmic rate in time.

However, the existence of growing mode solutions to the Klein--Gordon equation on the Reissner--Nordström-AdS spacetime constructed in this paper leads to the expectation of mode instability of the analogous Teukolsky equations, in contrast to the cases of Schwarzschild-AdS and Kerr-AdS. We will come back to this problem in our future work.

\subsubsection{Nonlinear (in)stability of the asymptotically AdS spacetime}
Now we discuss nonlinear (in)stability results for asymptotically AdS spacetimes. In \cite{holzegel2013stability}, Holzegel and Smulevici showed that, for any small spherical perturbation of the Schwarzschild-AdS initial data, under the Dirichlet boundary conditions, solutions of the Einstein--Klein--Gordon equations with a negative cosmological constant will converge to the Schwarzschild-AdS spacetime exponentially in time, demonstrating the spherical stability of the Schwarzschild-AdS spacetime. 

Despite the decay results for the Klein--Gordon equations \cite{holzegel2013decay,holzegel2014quasimodes} and the Teukolsky equations \cite{graf2024linear1,graf2024linear2,graf2024linear} mentioned in the previous section, this inverse logarithmic decay rate is believed to be too slow to ensure the nonlinear stability of the Kerr-AdS spacetime. Hence, in \cite{holzegel2013decay}, Holzegel and Smulevici made the following conjecture:
\begin{conjecture}[Holzegel--Smulevici, \cite{holzegel2013decay}]
The Kerr-AdS spacetimes are non-linearly unstable solutions to the initial-boundary value problem for the Einstein equations with Dirichlet boundary conditions.
\end{conjecture}
In contrast, the growing mode solutions we construct in this paper show instability even at the linear level. We naturally expect the nonlinear instability of the Reissner--Nordström-AdS spacetimes as solutions to the Einstein--Maxwell equations with a negative cosmological constant $\Lambda$ under reflecting boundary conditions.

\subsection{Comparing mechanisms for growing mode solutions to the Klein--Gordon equation on different spacetimes}

\subsubsection{The known instability mechanism I: Superradiant instability}
\label{previous work on mode instability}
Growing mode solutions on the Kerr and Kerr-AdS spacetimes violating the Hawking--Reall bound have been constructed in \cite{shlapentokh2014exponentially,dold2017unstable} respectively. Both spacetimes exhibit superradiance induced by rotation, namely, negative energy at the event horizon, which gives rise to the growing mode solutions\footnote{For the Klein--Gordon equation on Kerr, growing mode solutions are due to the combination of the superradiance and massive character of the equation. Heuristically, the Klein--Gordon mass serves as a reflecting mirror that would reflect the superradiance and result in a ``black-hole bomb''; see \cite{shlapentokh2014exponentially,damour1976quantum,detweiler1980klein,zouros1979instabilities}. Note on the contrary, the massless wave equation on the Kerr spacetime does not admit exponentially growing mode solutions \cite{shlapentokh2015quantitative,whiting1989mode}. Moreover, the solutions to the wave equation on Kerr are bounded and decay in time \cite{dafermos2010decay,dafermos2016decay,andersson2015hidden}.}. See \cite{shlapentokh2014exponentially,dold2017unstable} for a detailed discussion of this superradiant instability. See also \cite{press1972floating,damour1976quantum,detweiler1980klein,zouros1979instabilities} for some previous heuristic discussions.

While the rotation-induced superradiance phenomenon affects the Klein--Gordon equation on the Kerr/Kerr-AdS spacetimes, there is no such phenomenon on Reissner--Nordström/Reissner--Nordström-AdS since they are spherically symmetric. However, if one considers the charged Klein--Gordon equation on the charged spacetime, there is a charged analog of the superradiance induced by the coupling of the black hole charge $e$ and the scalar field charge $q_{0}$; see \cite{di2015superradiance,brito2020superradiance,denardo1973energetics,bachelot2004superradiance} and references therein. 
Growing mode solutions on Kerr-Newman-dS and Reissner--Nordström-dS have been constructed in \cite{besset2021existence}, for spacetimes where the product of the angular momentum and the Klein--Gordon mass is small compared to the product of the black hole charge $e$ and the scalar field charge $q_{0}$, which is a rigourous mathematical realization of this charge-induced superradiant instability. However, if the coupling is not strong, namely, the product of the black hole charge $e$ and the scalar field charge $q_{0}$ is small, then the exponential decay of the local energy is obtained in \cite{besset2020decay} for the Klein--Gordon equation on Reissner--Nordström-dS.

The Breitenlohner--Freedman bound $\alpha>-\frac{9}{4}$ is crucial to provide a positive scalar field energy near infinity. This is exploited to prove the local well-posedness of the asymptotically AdS spacetime; see \cite{breitenlohner1982stability} for the original physics argument and \cite{holzegel2012well,warnick2013massive} for a rigorous mathematics proof. The charged Klein--Gordon equation \eqref{Klein-Gordon} takes the form of the uncharged Klein--Gordon equation with a effective mass $$\alpha_{eff} = \alpha-\frac{l^{2}q_{0}^{2}A_{t}^{2}}{\Omega^{2}},$$ for which the relevant Breitenlohner--Freedman bound $\alpha_{eff}>-\frac{9}{4}$ is satisfied when $r$ approaches infinity. Hence, we still have the local well-posedness result. However, for the case of a large charge $q_{0}$, it is possible to violate the analogous condition of the Breitenlohner--Freedman bound\begin{equation*}
\alpha_{eff}\leq -\frac{9}{4}
\end{equation*}
for a bounded value of $r$. We exploit this fact in the proof of Lemma \ref{charged negative energy bound state}, which plays an important role in the construction. This superradiant mechanism has been identified as tachyonic instability in the context of physics literature; see, for instance \cite{brito2020superradiance,hartnoll2008holographic}.
\subsubsection{The new type of instability mechanism: Near-extremal instability}
\label{uncharged instability}
There is, however, no analogous superradiant mechanism for the uncharged Klein--Gordon equation on the Reissner--Nordström-AdS spacetime. Instead, the existence of growing mode solutions results from the black hole being too close to the extremality, which is called the near-extremal instability \cite{hartnoll2008holographic}.

We try to illustrate the mathematical meaning of this new type of instability by comparing the uncharged Klein--Gordon equations on Schwarzschild-AdS and Reissner--Nordström-AdS. Using the $(t,r,\theta,\varphi)$ coordinates as in \eqref{RN metric}, let $\Sigma_{t} = \{t = constant\}$ be the foliation of spacetimes. Putting the Dirichlet boundary conditions at timelike infinity, the standard energy estimate gives\begin{equation*}
\int_{\Sigma_{t_{1}}}T_{\mu\nu}X^{\mu}n^{\nu}dvol = \int_{\Sigma_{t_{2}}}T_{\mu\nu}X^{\mu}n^{\nu}dvol,
\end{equation*}
where $X = \partial_{t}$ is the unique timelike Killing vector field on the spacetimes, $n$ is the normal vector field on $\Sigma_{t}$, $dvol$ is the volume form of $\Sigma_{t}$, and the energy momentum tensor $T_{\mu\nu}$ for the Klein--Gordon equation is \begin{equation}
T_{\mu\nu} = \Re\left({\nabla_{\mu}\phi\overline{\nabla_{\nu}\phi}}\right)-\frac{1}{2}g_{\mu\nu}\left(g^{\alpha\beta}\nabla_{\alpha}\phi\overline{\nabla_{\beta}\phi}+\frac{\alpha}{l^{2}}\vert\phi\vert^{2}\right)
\end{equation}
The boundedness of the solution follows from a standard argument if the energy on the spacelike slice $\Sigma_{t = 0}$ is positive. The conservation of energy shows that the quantity \begin{equation*}
\int_{r_{+}}^{\infty}\int_{\mathbb{S}^{2}}\frac{T_{tt}}{\Omega^{2}}r^{2}\sin{\theta}drd\theta d\varphi
\end{equation*}
is independent of $t$. However, since the Klein--Gordon mass $\alpha$ in our setting is negative, the positivity of $T_{tt}$\begin{equation*}
T_{tt} \geq\frac{1}{2}\left(\vert\partial_{t} \phi\vert^{2}+\Omega^{2}\left(\Omega^{2}\vert\partial_{r}\phi\vert^{2}+\vert\slashed{\nabla}\phi\vert^{2}+\frac{\alpha}{l^{2}}\vert\phi\vert^{2}\right)\right)
\end{equation*}
is not guaranteed at first glance. On the Schwarzschild-AdS spacetime, the integral of the negative term $\frac{\alpha}{l^{2}}\vert\phi\vert^{2}$ can be absorbed by the integral of the derivative term $\Omega^{2}\vert\partial_{r}\phi\vert^{2}$ using the Hardy inequality. The following Hardy inequality can indeed be obtained:\begin{equation}
\label{nav hardy}
\int_{r_{+}}^{\infty}\int_{\mathbb{S}^{2}}\vert\phi\vert^{2}r^{2}\sin{\theta}drd\theta d\varphi\leq\frac{4}{9}\int_{r_{+}}^{\infty}\int_{\mathbb{S}^{2}}\vert\partial_{r}\phi\vert^{2}(r-r_{+})^{2}\left(r+r_{+}+\frac{r_{+}^{2}}{r}\right)^{2}drd\theta d\varphi.
\end{equation}
Using the above Hardy inequality, for the Schwarzschild-AdS spacetime, we have \begin{equation}
\begin{aligned}
&\int_{r_{+}}^{\infty}\int_{\mathbb{S}^{2}}\frac{T_{tt}}{\Omega^{2}}r^{2}\sin{\theta}drd\theta d\varphi = 
\int_{r_{+}}^{\infty}\int_{\mathbb{S}^{2}}\Omega^{2}\vert\partial_{r}\phi\vert^{2}r^{2}\sin\theta+\frac{\alpha}{l^{2}}\vert\phi\vert^{2}r^{2}\sin{\theta}drd\theta d\varphi\\\geq&
\int_{r_{+}}^{\infty}\int_{\mathbb{S}^{2}}(r-r_{+})\frac{1}{l^{2}r^{2}}\left(l^{2}r^{3}+r_{+}^{3}r^{2}+r_{+}^{4}r+r_{+}^{5}\right)\vert\phi\vert^{2}\sin{\theta}drd\theta d\varphi>0.
\end{aligned}
\label{after hardy}
\end{equation}
Moreover, for the Reissner--Nordström-AdS spacetime, using the argument as above, one can in fact show that if \begin{equation}
\frac{e^{2}}{r_{+}}<1+\frac{r_{+}^{2}}{l^{2}},\label{weakly charged}
\end{equation}
the energy is still positive. Hence, by the standard commuting vector field and redshift methods, one can still show the boundedness of the solutions. Combining \eqref{weakly charged} and \eqref{large}, we can show that growing mode solutions to the uncharged Klein--Gordon equation \eqref{Klein-Gordon} can only exist if\begin{equation}
\frac{e^{2}}{r_{+}^{2}}>1+\frac{1}{4\left(-\frac{3}{2}-\alpha\right)}.\label{large charge condition}
\end{equation}
In the above computation, the integral of the negative term $\frac{\alpha}{l^{2}} r^{2}\vert\phi\vert^{2}\sin{\theta}$ in the energy is controlled by the weighted $L^{2}$ norm of $\partial_{r}\phi$ by the Hardy inequality $\eqref{nav hardy}$. This weight in \eqref{nav hardy} decays quadratically like $(r-r_{+})^{2}$ towards the event horizon, whereas the weight $\Omega^{2}$ of the $L^{2}$ norm of $\partial_{r}\phi$ in the original energy of the Schwarzschild-AdS spacetime decays at a linear rate $C(r-r_{+})$ with $C$ bounded away from $0$. This fact is crucial in showing the absence of the negative energy. Hence, near the horizon, the integrand on the right-hand side of $\eqref{after hardy}$ is positive, giving the hope to absorb the negative term in the whole black hole exterior region $r>r_{+}$. However, on the extremal Reissner--Nordström-AdS spacetime, $\Omega^{2}$ decays exactly quadratically like $(r-r_{+})^{2}$. A simple computation shows that the analogous integrand on the right-hand side of $\eqref{after hardy}$ might be negative near the event horizon. Hence, for the near-extremal Reissner--Nordström-AdS spacetime, if the initial data $(\phi,\partial_{t}\phi){\bigl|_{t = 0}}$ is supported near the event horizon, the energy $\int_{\Sigma_{0}}T_{\mu\nu}X^{\mu}n^{\nu}$ might be negative; see Section \ref{sec:negative} for the proof of the existence of negative energy in near-extremal Reissner--Nordström-AdS spacetimes. This phenomenon is the key observation allowing us to construct the growing mode solution in our main theorem.

\subsection{Hairy black holes and connection to the AdS/CFT correspondence}
\label{physics}
\subsubsection{Connection between growing mode solutions, oscillating mode solutions, and hairy black holes}
The well-known no-hair conjecture in General Relativity asserts that all stationary black hole solutions to the Einstein (coupled with reasonable matter fields) equations are uniquely determined by the mass, charge, and angular momentum; for instance, see the review \cite{heusler1998stationary} and references therein.

In \cite{shlapentokh2014exponentially}, Shlapentokh-Rothman constructed solutions of the form $e^{-i\omega t}e^{im\varphi}S_{\kappa ml}(\theta)R(r)$ to the Klein--Gordon equation on the Kerr spacetime, which is interpreted as a ``linear hair'', in violation to at least the spirit of the conjecture. He showed the existence of growing mode solutions ($\omega$ has a negative imaginary part) and time-periodic solutions ($\omega\in\mathbb{R}$). Later, with Chodosh \cite{chodosh2017time}, they further constructed one-parameter families of the stationary axisymmetric asymptotically flat solutions $(\mathcal{M},g_{\delta},\phi_{\delta})$ to the Einstein--Klein--Gordon equations bifurcating off the Kerr solution. Their work gives a counter-example to the no-hair conjecture.\footnote{In fact, their result violates the no-hair conjecture to the extent that the metric is stationary but the scalar field is time-periodic.} See also \cite{herdeiro2014kerr,herdeiro2014ergosurfaces,brihaye2014myers,benone2014kerr,herdeiro2014non,cunha2015shadows} for previous physics and numerics results in this direction.

In the works \cite{van2021violent,li2023kasner}, the mathematical study of the hairy black hole interiors has been initiated. Putting the characteristic initial data $(g,\phi)$ on the event horizon, corresponding to a stationary non-zero scalar field $\phi$, Van de Moortel \cite{van2021violent} investigated the interior of the stationary Einstein--Maxwell--Klein--Gordon hairy black holes and their singular structures, assuming the scalar field is uncharged. Later, Van de Moortel and Li studied the analogous problem with the charged scalar field \cite{li2023kasner}. However, the existence of a stationary black hole exterior corresponding to the solutions in \cite{van2021violent,li2023kasner} was open; see Open Problem iv in \cite{van2021violent}. We resolve this problem in the companion work \cite{weihaozheng}, taking advantage of the method in this paper, which gives the affirmative answer to Open Problem iv and provides examples of asymptotically AdS hairy black holes whose interior is governed by the results of \cite{van2021violent,li2023kasner}.
\subsubsection{AdS/CFT correspondence, holographic superconductor, and near-horizon geometry}
Asymptotically AdS black holes have also emerged as an object of interest in the celebrated AdS/CFT correspondence \cite{maldacena1999large,maldacena2003eternal,witten1998anti}. Growing efforts have been spent in finding holographic analogs of superconductors, a problem for which hairy black holes appear as natural candidates. In fact, the existence of non-trivial stationary asymptotically AdS hairy black holes has been predicted in the paper \cite{hartnoll2008holographic}, using heuristics relying on the so-called near-horizon geometry. Indeed, it is possible to ``map'' the near-extremal Reissner--Nordström-AdS spacetime to $AdS_{2}\times\mathbb{S}^{2}$ by a limiting procedure. The heuristics of \cite{hartnoll2008holographic} argue that the $2$-dimensional Breitenlohner--Freedman bound is relevant to the existence of growing mode solutions to the uncharged Klein--Gordon equation on the extremal Reissner--Nordström-AdS spacetime, specifically for $\alpha$ in the range
\begin{equation*}
-\frac{9}{4}<\alpha<-\frac{3}{2}.
\end{equation*}
Our result actually shows the growing mode solutions exist within this range for sub-extremal Reissner--Nordström-AdS black holes.
We also identify explicitly a sufficient condition \eqref{large} for the growing mode solutions to exist, which is not previously identified in \cite{hartnoll2008holographic}. While our method in principle also works to construct growing mode solutions on the extremal Reissner--Nordström-AdS spacetime, we leave this question to future works.

\paragraph{Outline of the rest of the paper}
In Section $\ref{Preliminary}$, we introduce the geometry of the Reissner--Nordström-AdS spacetime and discuss the near-horizon boundary conditions and near-infinity boundary conditions. In Section $\ref{main theorem statement}$, we give precise statements of our main theorems and discuss the difficulties and new ideas. The construction of the stationary solutions of \eqref{Klein-Gordon} is viewed as an important step toward the construction of the growing mode solutions. We construct the stationary solution with arbitrary boundary conditions in Section $\ref{trival hair}$ and with reflecting boundary conditions in Section $\ref{proof of linear hair theorem}$ and $\ref{Neumann section}$, respectively. Lastly, in section $\ref{growing mode solution}$, we prove the existence of the desired growing mode solutions.
\section{Acknowledgements}
The author would like to thank his advisor, Maxime Van de Moortel, for his kind support, continuous encouragement, numerous inspiring discussions, and valuable suggestions for the manuscript. Additionally, the author expresses special thanks to Zheng-Chao Han for several very enlightening discussions. I would also like to thank two anonymous referees for their valuable suggestions.

\section{Preliminaries}
\label{Preliminary}
\subsection{Geometry of Reissner--Nordström-AdS space}
\label{geometry}
The Reissner--Nordström-AdS metric in $(t,r,\theta,\varphi)$ coordinates is given by
\begin{align*}
g_{RN} &= -\Omega^{2}dt^{2}+\frac{1}{\Omega^{2}}dr^{2}+r^{2}d\sigma^{2},\\
\Omega^{2} &= 1-\frac{2M}{r}+\frac{e^{2}}{r^{2}}+\frac{r^{2}}{l^{2}},
\end{align*}
where $M>0$ here is the mass of the black hole, $e\in\mathbb{R}$ is the charge of the black hole, and $\Lambda = -\frac{3}{l^{2}}<0$ is the negative cosmological constant. The Reissner--Nordström-AdS spacetime is uniquely determined by parameters $(M,e,l)$. We refer to these parameters $(M,e,l)$ as sub-extremal if the function $\Omega^{2}$ admits two distinct positive roots, denoted as $0<r_{-}<r_{+}$. The region $r>r_{+}$ is known as the exterior of the Reissner--Nordström-AdS black hole, and the hypersurface $\{r = r_{+}\}$ is called the event horizon. We can define the horizon temperature $T$ as: \begin{equation}
\label{temperature}
T: = \frac{d\Omega^{2}}{dr}(r_{+}).
\end{equation} 
For sub-extremal parameters, since $r_{+}$ is the largest real root, we have $T>0$. Parameters $(M,e,l)$ are called extremal if $\Omega^{2}$ has two identical positive roots. For extremal Reissner--Nordström-AdS, we have $T = 0$. 

Note that the Reissner--Nordström-AdS spacetime breaks down in the usual $(t,r,\theta,\varphi)$ coordinates when $r = r_{+}$. To extend the spacetime smoothly to the event horizon, let\begin{equation*}
r^{*}(r) = \int_{\infty}^{r}\left(1-\frac{2M}{r}+\frac{e^{2}}{r^{2}}+\frac{r^{2}}{l^{2}}\right)^{-1}dr.
\end{equation*}
we can see that\begin{align*}
\lim_{r\rightarrow\infty}r^{*}(r) &= 0,\\
\lim_{r\rightarrow r_{+}}r^{*}(r) &= -\infty.
\end{align*}
Then we can define the outgoing Eddington--Finkelstein (EF) coordinates by\begin{equation*}
v = t+r^{*},\quad r = r.
\end{equation*}
The Reissner--Nordström-AdS metric in the outgoing EF coordinates $(v,r,\theta,\varphi)$ becomes\begin{equation*}
g_{RN} = -\Omega^{2}dv^{2}+2dvdr+r^{2}d\sigma^{2}.
\end{equation*}
As we can see from the asymptotic behavior of $r^{*}(r)$ when $r$ approaches infinity, the affine length of the curve with constant $(t,\theta,\varphi)$ is finite, contrary to the asymptotically flat Reissner--Nordström spacetime, where $r^{*}$ there approaches infinity. By letting $g = \Omega^{2}\widetilde{g}$, we can show that the boundary of $(\mathcal{M},\widetilde{g})$ with coordinates $(t,r^{*} = 0,\theta,\varphi)$ is timelike; see Figure $\ref{fig:penrose}$ for the Penrose diagram of the Reissner--Nordström-AdS spacetime.
\begin{figure}[htbp]
    \centering
    \begin{tikzpicture}[scale=2]
        \draw[dashed] (-1,-1) node[left] {$i^{+}$} -- (0,-2);
        \draw[dashed] (1,-1) node[right] {$i^{+}$} -- (0,-2);
        \draw[dashed] (-1,1) node[left] {$i^{+}$} -- (0,2);
        \draw[dashed] (1,1) node[right] {$i^{+}$} -- (0,2);
        
        \draw[dashed] (-1,1) -- (1,-1);
        \draw[dashed] (-1,-1) -- (1,1);
        \draw (-1,2) -- (-1,-2);
        \draw (1,2) -- (1,-2);
        \node at (1.3,0) {$r = \infty$};
        \node at (-1.3,0) {$r = \infty$};
        \node at (1.3,1.3) {$r = 0$};
        \node at (-1.3,1.3) {$r = 0$};
        \node at (1.3,-1.3) {$r = 0$};
        \node at (-1.3,-1.3) {$r = 0$};
        \node at (0.6,0.3) {$r = r_{+}$};
        \node at (0.6,-0.3) {$r = r_{+}$};
        \node at (-0.6,0.3) {$r = r_{+}$};
        \node at (-0.6,-0.3) {$r = r_{+}$};
        \node at (0.5,1.2) {$r = r_{-}$};
        \node at (-0.5,1.2) {$r = r_{-}$};\node at (0.5,-1.2) {$r = r_{-}$};\node at (-0.5,-1.2) {$r = r_{-}$};
        
    \end{tikzpicture}
    \caption{Penrose diagram for the asymptotically AdS spacetime.}
    \label{fig:penrose}
\end{figure}
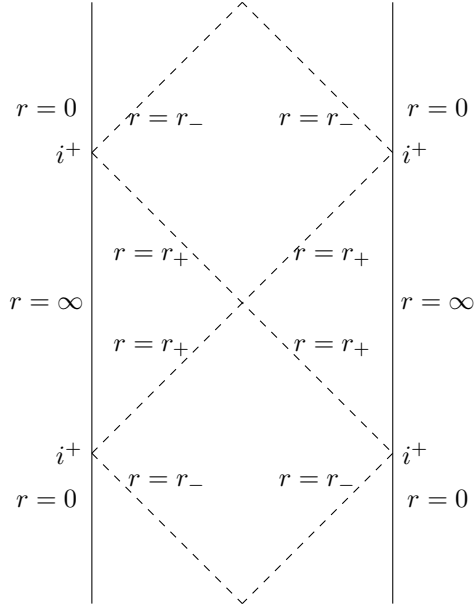

\subsection{Parameters for the spherically symmetric Klein--Gordon equation}
\label{section:parameters}
The spherically symmetric Klein--Gordon equation on the Reissner--Nordström-AdS background is \begin{equation}
-\frac{r^{2}}{\Omega^{2}}\frac{\partial^{2}\phi}{\partial t^{2}}-2iq_{0}r^{2}\frac{A_{t}}{\Omega^{2}}\partial_{t}\phi+\frac{\partial}{\partial r}\left(r^{2}\Omega^{2}\frac{\partial\phi}{\partial r}\right) = \left(\frac{1}{l^{2}}\alpha-\frac{q_{0}^{2}A_{t}^{2}}{\Omega^{2}}\right)r^{2}\phi.\label{Klein-Gordon on RN}
\end{equation}
If we further assume $\phi$ is stationary, then we can get the stationary Klein--Gordon equation:\begin{equation}
\label{stationary Klein-Gordon}
\frac{d}{dr}\left(r^{2}\Omega^{2}\frac{d\phi}{dr}\right) = \left(\frac{\alpha}{l^{2}}-\frac{q_{0}^{2}A_{t}^{2}}{\Omega^{2}}\right)r^{2}\phi.
\end{equation}
We call $(M,e,l,\alpha,q_{0})$ the parameters of the Klein--Gordon equation $\eqref{Klein-Gordon on RN}$ and we call them sub-extremal parameters if the corresponding $(M,e,l)$ are sub-extremal parameters for Reissner--Nordström-AdS spacetime. However, since the sub-extremal condition concerns the roots of $\Omega^{2}$, which are solutions of a quartic equation, it is hard to provide an algebraic criterion for when the parameters will satisfy the sub-extremal condition. Hence, it may be inconvenient to use $(M,e,l,\alpha,q_{0})$ as parameters. In the following proposition, we prove that we can always use $(M,r_{+},l,\alpha,q_{0})$ as parameters.
\begin{proposition}
The parameters transformation $(M,e,l,\alpha,q_{0})\rightarrow(M,r_{+},l,\alpha,q_{0})$ is a regular transformation. The new parameters $(M,r_{+},l,\alpha,q_{0})$ satisfy the sub-extremal condition if and only if\begin{align}
&M\geq M_{e = 0} :=\frac{r_{+}}{2}\left(1+\frac{r_{+}^{2}}{l^{2}}\right) ,\label{subextremal1}\\&
M<M_{0} := r_{+}\left(1+2\frac{r_{+}^{2}}{l^{2}}\right).\label{subextremal2}
\end{align}
In other words, for fixed $(r_{+},l,\alpha,q_{0})$, the sub-extremal condition can be achieved by decreasing $M$ to be smaller than $M_{0}$ but larger than $M_{e = 0}$.
\end{proposition}
\begin{proof}
By the form of $\Omega^{2}$, we have\begin{equation*}
e^{2} = -r_{+}^{2}\left(1-\frac{2M}{r_{+}}+\frac{r_{+}^{2}}{l^{2}}\right)
\end{equation*}
It is easy to check the Jacobian of the map $$(M,r_{+},l,\alpha,q_{0})\rightarrow(M,e,l,\alpha,q_{0})$$ is non-zero. Hence, the transformation is regular. By condition $\eqref{subextremal1}$, we have $e^{2}\geq0$. 

Calculating the surface gravity $\kappa$ of the black hole, we can get\begin{align*}
\kappa =& \frac{2M}{r_{+}^{2}}-\frac{2e^{2}}{r_{+}^{3}}+2\frac{r_{+}}{l^{2}}\\=&
\frac{2}{r_{+}^{2}}(M_{0}-M)>0
\end{align*}
by condition \eqref{subextremal2}. Since \begin{equation*}
\Omega^{2} = \frac{\frac{1}{l^{2}}r^{4}+r^{2}-2Mr+e^{2}}{r^{2}},
\end{equation*}
by Vieta's formula, we have \begin{align}
&\sum_{i = 1}^{4}r_{i} = 0,\label{Vieta sum}\\&
\prod\limits_{1}^{4} r_{i} = l^{2}e^{2}>0,\label{Vieta product}
\end{align}
where $r_{i}$, $i = 1,\dots,4$ are the four roots of the polynomial $\frac{1}{l^{2}}r^{4}+r^{2}-2Mr+e^{2}$.

Given that $\Omega^{2}$ already has a positive root $r_{+}$, either it has two real roots $r_{1} = r_{+}$ and $r_{2} = r_{-}$ along with two non-real conjugate roots $r_{3}$ and $r_{4}$, or it has four real roots $r_{1}>r_{2}>r_{3}>r_{4}$. For the first case, we have $r_{3}\times r_{4}>0$. Thus we have $r_{1}>0$, $r_{2}>0$ by $\eqref{Vieta product}$. Since $\kappa>0$ by condition $\eqref{subextremal2}$, $r_{+}$ is the largest root and $(M,e,l)$ are sub-extremal. If $\Omega^{2}$ has four real roots, then by $\eqref{Vieta sum}$ and $\eqref{Vieta product}$, we have $r_{1},\ r_{2}$ positive and $r_{3},\ r_{4}$ negative. Since $r_{+}>0$ and $\kappa>0$, then $r_{+}$ is the largest root and parameters are sub-extremal.
\end{proof}

\subsection{Boundary conditions}
\label{sec: boundary conditions}
In this section, we give the precise definition of Dirichlet, Neumann, and Robin boundary conditions. Let $\Delta: = \sqrt{\frac{9}{4}+\alpha}$.
\begin{definition}[\textup{\cite{warnick2013massive}}]
We say a $C^{1}$ function $f$ on $\mathcal{M}$ obeys Dirichlet, Neumann, or Robin boundary conditions if the following holds:
\begin{enumerate}
  \item[(1)] Dirichlet:\begin{equation*}
  r^{\frac{3}{2}-\Delta}f\rightarrow 0,\quad r\rightarrow \infty.
  \end{equation*}
  \item[(2)] Neumann:\begin{equation*}
  r^{2\Delta+1}\frac{d}{dr}\left(r^{\frac{3}{2}-\Delta}f\right) \rightarrow 0,\quad r\rightarrow\infty.
  \end{equation*}
  \item[(3)] Robin:\begin{equation*}
  r^{2\Delta+1}\frac{d}{dr}\left(r^{\frac{3}{2}-\Delta}\phi\right)+\beta r^{\frac{3}{2}-\Delta}\phi\rightarrow 0,\quad r\rightarrow\infty,
  \end{equation*}
  where $\beta$ is a real constant.
\end{enumerate}
\end{definition}

\subsection{Mode solutions}
Assume $\phi = e^{i\omega t}\psi(r)$ is a spherically symmetric mode solution of $\eqref{Klein-Gordon on RN}$ satisfying the Dirichlet boundary condition, then we can get the equation for $\psi$:\begin{equation}
\label{Klein-Gordon equation for psi}
\frac{d}{dr}\left(r^{2}\Omega^{2}\frac{d\psi}{dr}\right) = \left(\frac{1}{l^{2}}\alpha-\frac{q_{0}^{2}A_{t}^{2}}{\Omega^{2}}-\frac{\omega^{2}}{\Omega^{2}}-\frac{2q_{0}A_{t}\omega}{\Omega^{2}}\right)r^{2}\psi.
\end{equation}
To facilitate the discussion of solutions under Neumann boundary conditions, letting $\psi = r^{-\frac{3}{2}+\Delta}\Psi$, we can rewrite the equation $\eqref{Klein-Gordon equation for psi}$ as:\begin{equation}
\label{twisted Klein-Gordon equation for psi}
\begin{aligned}
&\frac{d}{dr}\left(r^{2\Delta-1}\Omega^{2}\frac{d\Psi}{dr}\right) \\=&
r^{2\Delta-1}\left(\frac{1}{l^{2}}\alpha-\frac{q_{0}^{2}A_{t}^{2}}{\Omega^{2}}-\frac{\omega^{2}}{\Omega^{2}}-\frac{2q_{0}A_{t}\omega}{\Omega^{2}}+(\frac{3}{2}-\Delta)(-\frac{1}{2}+\Delta)\frac{\Omega^{2}}{r^{2}}+(\frac{3}{2}-\Delta)\frac{1}{r}\frac{d\Omega^{2}}{dr}\right) \Psi.
\end{aligned}
\end{equation}
From the asymptotic analysis, the behavior of $\psi$ near $r = \infty$ is given by Dirichlet, Neumann, or Robin conditions. Next, we derive the asymptotic behavior of $\psi$ when $r\rightarrow r_{+}$. Since the Reissner--Nordström-AdS metric can be extended to the event horizon in the outgoing Eddington--Finkelstein coordinates $(v = t+r^{*},r,\theta, \varphi)$, we require the solution $\phi$ of the Klein--Gordon equation can also be extended to the event horizon. We have\begin{equation*}
\phi = e^{i\omega t}\psi(r) = e^{i\omega v}e^{-i\omega r^{*}}\psi(r).
\end{equation*}
Hence to do the extension, we require that $\rho(r): = e^{-i\omega r^{*}}\psi$ can be extended to a smooth function on $[r_{+},\infty)$. In other word, we require that for any $j\geq0$,  $\frac{d^{j}\rho}{dr^{j}}$ has a limit when $r\rightarrow r_{+}$. We have the following relation\begin{align}
\psi &= e^{i\omega r^{*}}\rho(r),\label{function relation}\\
\psi&\approx Ce^{i\omega r^{*}},\quad r\rightarrow r_{+}.\label{asymptotic behavior of psi near horizon under rstar coordintaes}
\end{align}
Using $r$ as the function variable, from \eqref{asymptotic behavior of psi near horizon under rstar coordintaes}, for spherically symmetric growing mode solutions with $\Im(\omega)<0$, we have the asymptotic decay behavior:
\begin{equation}
\vert\psi\vert\approx \vert r-r_{+}\vert^{-\frac{1}{\kappa}\Im(\omega)}.\label{decay of growing mode near horizon}
\end{equation}
Differentiating $\eqref{function relation}$, we have \begin{equation}
\frac{d\psi}{dr^{*}} = i\omega\psi+O(r-r_{+}).\label{near horizon for psi}
\end{equation}
 
We can prove the following result:
\begin{proposition}
For any given sub-extremal parameters $(M,r_{+},l,\alpha,q_{0})$, the only real spherically symmetric mode solution to the Klein-Gordon equation $\eqref{Klein-Gordon equation for psi}$ with Dirichlet or Neumann boundary conditions which can be extended smoothly to the horizon is the static solution $\phi(t,r) = \psi(r)$ with $\omega = 0$.
\end{proposition}
\begin{proof}
Multiplying $\overline{\psi}$ and taking the imaginary part of the equation $\eqref{Klein-Gordon equation for psi}$, we have\begin{equation*}
\Im\left(\frac{d}{dr}\left(r^{2}\Omega^{2}\frac{d\psi}{dr}\right)\overline{\psi}\right) = r^{2}\vert\psi\vert^{2}\Im\left(\frac{1}{l^{2}}\alpha-\frac{q_{0}^{2}A_{t}^{2}}{\Omega^{2}}-\frac{\omega^{2}}{\Omega^{2}}-\frac{2q_{0}A_{t}\omega}{\Omega^{2}}\right).
\end{equation*}
Hence if $\omega$ is real, we have\begin{equation}
\frac{d}{dr}\Im\left(r^{2}\Omega^{2}\overline{\psi}\frac{d\psi}{dr}\right) = 0.
\end{equation}
For $\psi$ satisfying the Dirichlet boundary condition, using $\eqref{near horizon for psi}$ and $\eqref{function relation}$, we have
\begin{equation}
\lim_{r\rightarrow\infty}\Im\left(r^{2}\Omega^{2}\bar{\psi}\frac{d\psi}{dr}\right) = 0
\end{equation}
Hence we have
\begin{equation}
r_{+}^{2}\lim_{r\rightarrow r_{+}}\Im\left(\Omega^{2}\overline{\psi}\frac{d\psi}{dr}\right) = r_{+}^{2}\omega\vert \rho\vert^{2}(r_{+}) =\lim_{r\rightarrow\infty}\Im\left(r^{2}\Omega^{2}\bar{\psi}\frac{d\psi}{dr}\right) = 0.
\end{equation}
Noting that if $\rho(r_{+})= 0$, the linear ODE theory to \eqref{stationary Klein-Gordon} will imply $\phi\equiv0$. Hence $\rho(r_{+})\neq 0$ and $\omega = 0$. For $\psi$ satisfying the Neumann boundary condition, using the equation $\eqref{twisted Klein-Gordon equation for psi}$ and reapplying the above argument, we can prove the same conclusion.
\end{proof}
Furthermore, for the uncharged Klein--Gordon equation, we have the following result about growing mode solutions.
\begin{proposition}
\label{imaginary omega pro}
All spherically symmetric growing mode solutions to the uncharged equation \eqref{Klein-Gordon equation for psi} on the sub-extremal Reissner--Nordström-AdS spacetime have pure imaginary mode $\omega$.
\end{proposition}
\begin{proof}
Assume $\phi = e^{i\omega t}\psi$ is a growing mode solution to the uncharged Klein--Gordon equation with $\omega = \omega_{R}+i\omega_{I}$. Then multiplying $\overline{\psi}$ and taking the imaginary part of the equation \eqref{Klein-Gordon equation for psi}, we have\begin{equation*}
\Im\left(\frac{d}{dr}\left(r^{2}\Omega^{2}\frac{d\psi}{dr}\right)\bar{\psi}\right) = -r^{2}\vert\psi\vert^{2}\frac{2\omega_{R}\omega_{I}}{\Omega^{2}}.
\end{equation*}
Integrating the above equations and using the boundary condition \eqref{decay of growing mode near horizon}, we have\begin{equation}
0 = -2\omega_{R}\omega_{I}\int_{r_{+}}^{\infty}\frac{r^{2}\vert\psi\vert^{2}}{\Omega^{2}}\mathrm{d}r.
\end{equation}
Since $\omega_{I}$ is negative for spherically symmetric growing mode solutions, we conclude that $\omega_{R}$ is zero.
\end{proof}

\section{Main results and outline of the proof}
\label{main theorem statement}
\subsection{Precise statements of main results}
Now we are ready to state the main results we get.
\begin{theorem}
\label{growing mode theorem}
For Klein--Gordon equation $\eqref{Klein-Gordon}$ with negative mass $\alpha$, let $C_{DN} = 0$ for Dirichlet boundary conditions and $C_{DN} = -\frac{5}{4}$ for Neumann boundary conditions. Imposing reflecting boundary condition for \eqref{Klein-Gordon}, we have\begin{enumerate}
\item[(1)]\textup{(Large charge case)} For any given sub-extremal parameters $(M_{b},r_{+},l)$ with $M_{e = 0}<M_{b}<M_{0}$ and $-\frac{9}{4}<\alpha<C_{DN}$, there exists a $q_{1}>0$, such that for any $\vert q_{0}\vert>q_{1}$ and $\delta$ sufficiently small, there exist real analytic functions $M_{e = 0}<M(\epsilon)<M_{b}$ and $\omega_{R}(\epsilon)\in\mathbb{R}$ with $\omega_{R}(0) = 0$ on $-\delta<\epsilon<\delta$ so that there exists a mode solution $\phi = e^{i(\omega_{R}(\epsilon)+i\epsilon)t}\psi(r)$ of $\eqref{Klein-Gordon}$ with sub-extremal parameters $\left(M(\epsilon),r_{+},l,\alpha,q_{0}\right)$. Moreover, $\phi$ can be continuously extended to the event horizon $r = r_{+}$.

\item[(2)]\textup{(General fixed charge case)} For each fixed parameters $(r_{+},l,\alpha,q_{0})$ satisfying\begin{align}
&-\frac{9}{4}<\alpha<\left\{-\frac{3}{2}+\frac{1}{2}l^{2}q_{0}^{2},C_{DN}\right\},\label{hardy assumption}\\
&\frac{r_{+}^{2}}{l^{2}}>R_{0},\label{large condition}
\end{align}
where $R_{0}$ is the positive solution to the quadratic equation \eqref{quadratic}. Then for any $\delta$ sufficiently small, there exist real analytic functions $M_{e = 0}<M(\epsilon)<M_{0}$ and $\omega_{R}(\epsilon)\in\mathbb{R}$ with $\omega_{R}(0) = 0$ on $-\delta<\epsilon<\delta$, such that $(M(\epsilon),r_{+},l)$ are sub-extremal parameters and there exists a mode solution $\phi = e^{i(\omega_{R}(\epsilon)+i\epsilon)t}$ of $\eqref{Klein-Gordon}$ with parameters $\left(M(\epsilon),r_{+},l,\alpha,q_{0}\right)$. Moreover, $\phi$ can be continuously extended to the event horizon $r = r_{+}$ and we have\begin{align}
&\frac{dM}{d\epsilon}(0)<0,\label{deri of M}\\
&q_{0}e\frac{d\omega_{R}}{d\epsilon}(0)<0,\quad q_{0}e\neq 0.\label{deri of omega}
\end{align}
\item[(3)]\textup{(Weakly charged case)} Let fixed parameters $(r_{+},l,\alpha,q_{0} = 0)$ satisfying\begin{align}
&-\frac{9}{4}<\alpha<-\frac{3}{2},\\&
\frac{r_{+}^{2}}{l^{2}}>\frac{1}{4\left(-\frac{3}{2}-\alpha\right)},
\end{align}
there exists $M_{e = 0}<M_{c}<M_{0}$ such that for any $M\in(M_{c},M_{0})$, there exists a growing mode solution to \eqref{Klein-Gordon} with parameters $(M,r_{+},l,q_{0} = 0)$. Furthermore, for $M\in(M_{c},M_{0})$, there exists $\delta>0$ depending on $M$, such that for all $0\leq\vert q_{0}\vert\leq\delta$, \eqref{Klein-Gordon} with parameters $(M,r_{+},l,\alpha,q_{0})$ has a growing mode solution.
\end{enumerate}
\end{theorem}
\begin{remark}
In both the large charge case and the general fixed charge case, when $\epsilon = 0$, the solutions constructed in the theorem reduce to stationary solutions. One should think the growing mode solutions in these two cases are constructed by perturbing a stationary solution. However, in view of Proposition \ref{imaginary omega pro}, one can treat the existence of growing mode solutions to the uncharged Klein--Gordon equation \eqref{Klein-Gordon equation for psi} as a negative eigenvalue problem. This perspective allows us to construct growing mode solutions to the uncharged Klein--Gordon equation for all $M\in(M_{c}, M_{0})$ without using a perturbative approach. Additionally, growing mode solutions for the weakly charged case can be obtained by perturbing the scalar field charge.
\end{remark}
\begin{remark}
In contrast to Proposition \ref{imaginary omega pro}, \eqref{deri of omega} shows the existence of an oscillating profile ($\Re\omega\neq0$) for growing mode solutions to the charged Klein--Gordon equations.
\end{remark}

The key step toward proving the existence of spherically symmetric growing mode solutions is the construction of non-trivial bounded stationary spherically symmetric solutions to $\eqref{Klein-Gordon}$ with Dirichlet or Neumann boundary conditions. Under the stationary and spherically symmetric assumption, $\eqref{Klein-Gordon}$ is reduced to the ODE $\eqref{stationary Klein-Gordon}$. Instead of solving ODE $\eqref{stationary Klein-Gordon}$ under given boundary conditions, we start with regular data at $r = r_{+}$ and solve the equation $\eqref{stationary Klein-Gordon}$ to $r = \infty$. Considering the asymptotic behavior of the solution, we have the following proposition, the proof of which will be postponed to Section \ref{trival hair}:
\begin{proposition}
\label{trival theorem}
For any sub-extremal parameters $(M,r_{+},l,\alpha,q_{0})$ satisfying the Breitenlohner--Freedman bound $-\frac{9}{4}<\alpha<0$, we can construct a static spherically symmetric solution $\phi(r)$ to the Klein--Gordon equation $\eqref{Klein-Gordon}$ on a fixed Reissner--Nordström AdS metric $g$ such that\begin{align}
&\lim_{r\rightarrow r_{+}}\phi(r) = 1,\\
&\phi(r) = C_{D}u_{D}(r)+C_{N}u_{N}(r).
\end{align}
Here $\{u_{D}(r),u_{N}(r)\}$ is the local basis of the solution of equations $\eqref{Klein-Gordon on RN}$, with the following asymptotic behaviors:\begin{align}
\label{Dirichlet asymptotic}
&\lim_{r\rightarrow\infty}r^{\frac{3}{2}+\Delta}u_{D}(r) = 1,\quad \lim_{r\rightarrow\infty}r^{\frac{5}{2}+\Delta}\frac{du_{D}}{dr} = -\frac{3}{2}-\Delta,
\\&\lim_{r\rightarrow\infty}r^{\frac{3}{2}-\Delta}u_{N}(r) = 1,\quad \lim_{r\rightarrow\infty}r^{\frac{5}{2}-\Delta}\frac{du_{N}}{dr} = -\frac{3}{2}+\Delta.\label{Neumann asymptotic}
\end{align}
In other words, $u_{D}$ is a function on $\mathcal{M}$ satisfying the Dirichlet boundary condition and $u_{N}$ is a function on $\mathcal{M}$ satisfying the Neumann boundary condition.
\end{proposition}
\begin{remark}
In the later discussion, we use the notation $u_{D}\approx r^{-\frac{3}{2}-\Delta}$ to denote the asymptotic behavior $\eqref{Dirichlet asymptotic}$ and notation $u_{N}\approx r^{-\frac{3}{2}+\Delta}$ to denote the asymptotic behavior $\eqref{Neumann asymptotic}$. 
\end{remark}

The above theorem holds for any sub-extremal parameters $(M,r_{+},l,\alpha,q_{0})$, highlighting the significance of the boundary conditions. However, for a given boundary condition, it is highly non-trivial that one can prove the existence of a stationary solution on $[r_{+},\infty)$ which is regular at $r = r_{+}$ and satisfies the given boundary condition.

Furthermore, we can prove the following existence of non-trivial stationary solutions to $\eqref{Klein-Gordon}$ with reflecting boundary conditions.
\begin{theorem}
\label{linear hair theorem}
Imposing the reflecting boundary condition for the Klein--Gordon equation $\eqref{Klein-Gordon}$ with a negative mass $\alpha\in(-\frac{9}{4},C_{DN})$, we have\begin{enumerate}
\item[(1)]\textup{(Large charge scalar field)} For any given sub-extremal parameters $(M_{b},r_{+},l)$ with $M_{e = 0}<M_{b}<M_{0}$ and $\alpha$ within the above range, there exists a $q_{1}>0$, such that for any $\vert q_{0}\vert>q_{1}$, there exists $M_{e = 0}<M = M(M_{b},r_{+},l,\alpha,q_{0})<M_{b}$, such that a stationary solution to $\eqref{Klein-Gordon}$ with parameters $(M,r_{+},l,\alpha,q_{0})$ exists and can be extended continuously to the event horizon $\{r = r_{+}\}$.
\item[(2)]\textup{(General charged scalar field)} For any given parameters $(r_{+},l,\alpha,q_{0})$ satisfying the conditions \eqref{hardy assumption} and \eqref{large condition}, there exists $M_{e = 0}<M<M_{0}$ such that a stationary solution $\phi$ to $\eqref{Klein-Gordon}$ with parameters $(M,r_{+},l,\alpha,q_{0})$ exists and can be extended continuously to the event horizon $\{r = r_{+}\}$.
\end{enumerate}
\end{theorem}

\subsection{Main difficulties and outline of the proof}
\label{outline of the proof}
In this section, we discuss the main difficulties in the proof and ideas we used by considering the uncharged Klein--Gordon equation with Dirichlet boundary conditions; see Section \ref{outline of neumann} for the discussion of Neumann boundary conditions. The ideas used in treating the charged case are essentially the same. We mainly adopt the method outlined in \cite{shlapentokh2014exponentially}. The proof is divided into two main steps:\begin{itemize}
\item Prove the existence of stationary solutions to Klein--Gordon equation \eqref{Klein-Gordon} as in Theorem \ref{linear hair theorem}.
\item Perturb the stationary solution and use the implicit function theorem to get growing mode solutions.
\end{itemize}
The main difficulty is the first step. To show the existence of a non-trivial stationary solution $\phi$ with Dirichlet boundary condition, we apply the variational method to find the minimizer of the energy functional of $\eqref{stationary Klein-Gordon}$:\begin{equation*}
L[f] = \int_{r_{+}}^{\infty}r^{2}\Omega^{2}\left(\frac{df}{dr}\right)^{2}+\frac{1}{l^{2}}\alpha r^{2}f^{2}dr.
\end{equation*}
The asymptotic analysis of solution $\phi$ gives that $L[\phi]$ is finite if and only if $\phi$ is bounded and satisfies the Dirichlet boundary condition. Furthermore, multiplying the equation $\eqref{stationary Klein-Gordon}$ by $\phi$ and applying the integration by parts, one can get $L[\phi] = 0$ for the stationary solution $\phi$.

To construct the desired solution, we apply the variational argument to $L[f]$ within the function class $H_{0}^{1}$ with $$\int_{r_{+}}^{\infty}\frac{r^{2}}{\Omega^{2}}f^{2}\mathrm{d}r= 1.$$By the constrained variational principle, the minimizer $\phi$ will be the solution of the following eigenvalue equation \begin{equation*}
-\frac{d}{dr}\left(r^{2}\Omega^{2}\frac{d\phi}{dr}\right)+\frac{1}{l^{2}}\alpha r^{2}\phi = -\lambda(M,r_{+},l,\alpha)\frac{r^{2}}{\Omega^{2}}\phi.
\end{equation*}
For fixed neutral parameters $(r_{+},l,\alpha,q_{0} = 0)$, we denote the $M$-dependence of $\phi$ by $\phi_{M}$. To show $\lambda(M)$ can be zero for some $M$, we start with $M = M_{0}$ such that $(M,r_{+},l)$ are extremal parameters for the given $(r_{+},l)$ and show that $\lambda(M_{0})$ is negative. Then by exploiting the monotonicity and continuity of $L[f]$ with respect to $M$, we conclude the existence of $M = M_{c}$ such that $\lambda(M_{c}) = 0$.

However, several difficulties arise when attempting to apply the above framework. 
\begin{enumerate}
  \item First, the energy functional is not even lower bounded at first glance, as the functions in our function class essentially have an $L^{2}$ norm equal to $1$ while $L[f]$ contains the integral of $r^{2}f^{2}$. 
  \item Second, even if $L[f]$ is lower bounded in our function class, since $\Omega^{2}$ degenerates at the event horizon, we could not get the $H^{1}$ boundedness of the minimizing sequence $\phi_{M,n}$. Hence, to the best of our knowledge, no (weak) convergence result we can use. 
  \item Third, the above framework requires us to start with $M = M_{0}$, and the corresponding $L_{M_{0}}[f]$ should have a negative minimum. However, as already shown in Section \ref{uncharged instability}, for small black hole charge \eqref{weakly charged}, one can prove the positivity of the energy functional, which means the existence of the energy functional with a negative minimum is subtle in our setting.
\end{enumerate}

The second difficulty is the easiest one to overcome. We can directly use the method outlined in \cite{shlapentokh2014exponentially} by considering the perturbed energy functional $L^{\epsilon}[f]$, which has the same integrand but is integrated over $(r_{+}+\epsilon,\infty)$. Then for each fixed $\epsilon$, we can show that $\Omega^{2}$ is lower bounded on $(r_{+}+\epsilon,\infty)$. Thereby, we can prove the $H^{1}$ boundedness of the minimizing sequence of this perturbed energy functional. By the weak compactness of $H^{1}$ and Sobolev embedding, we can show the existence of the minimizer $\phi^{\epsilon}$ of $L^{\epsilon}[f]$; See Proposition \ref{eigenvalue}. To get rid of $\epsilon$, our method is to take a step back by showing the local $H^{1}$ boundedness on any fixed compact set $K\subset (r_{+},\infty)$ and then achieving local convergence for each $K$. Nonetheless, the limit obtained from the local convergence argument might be trivial since the energy can concentrate outside of the compact set $K$ despite having $\Vert f\Vert_{L^{2}(r_{+}+\epsilon,\infty)} = 1$. Therefore, we still need some coercivity results.

One of the main ideas in this paper lie in the resolution of the first difficulty and coercivity mentioned above, which uses the so-called twisted derivative obtained by replacing the usual derivative $\frac{d}{dr}$ with $\widetilde{\nabla}_{r}(\cdot) = h\frac{d}{r}(h^{-1}\cdot)$; see Section \ref{sec:energy functional}. The twisted derivative has been used to establish the local well-posedness of the Klein--Gordon equation on asymptotically AdS space with Neumann boundary conditions in \cite{warnick2013massive}. This approach addresses the difficulty arising from the fact that the energy functional $L[f]$ is infinite for functions with Neumann boundary conditions. The surprising aspect here is that, even for Dirichlet boundary conditions, we need to employ the twisted derivative. Writing the equation and energy functional in the twisted derivative form, the structure of $\Omega^{2}$ will influence the potential term, enriching the sign structure of the potential. By a careful analysis of the potential and integrability of solutions $\phi^{\epsilon}$ in Proposition \ref{nontrival solution} and \ref{gain regularity} we can get the lower boundedness and coercivity.

Dealing with the third difficulty of achieving a negative minimum for the uncharged $L[f]$ is particularly challenging. To overcome this obstacle, we derive a sharp near-horizon version of the Hardy inequality in Lemma \ref{hardy} and devise a test function with compact support to demonstrate the existence of a negative minimum in Lemma \ref{nonempty}.

The non-trivial stationary solution with Neumann boundary conditions can be constructed similarly, by using a different twist function $h$ in the twisted derivative.

Growing mode solutions can be obtained from the stationary solution through the application of the implicit function theorem; see Section \ref{growing mode solution}.
\begin{remark}
One should also note that the first difficulty and second difficulty can be overcome by using the abstract framework proved in \cite{holzegel2014boundedness}, where a good coordinate system on general asymptotically AdS spacetimes was chosen and the energy functional $L[f]$ can be written as an integral in a compact region with a degenerated coefficient. To overcome this degeneracy, the approach in \cite{holzegel2014boundedness} was to establish a Sobolev theory adapted to the degeneracy, using the idea of twisted derivatives. The existence of the minimizer to $L[f]$ then follows from the standard variational method. We follow the framework of \cite{shlapentokh2014exponentially} in this paper and provide a self-contained proof by using twisted derivatives.
\end{remark}
\section{Proof of Proposition $\ref{trival theorem}$}
\label{trival hair}
The proof of Proposition $\ref{trival theorem}$ is a standard application of the asymptotic analysis. We provide a proof in this section for completeness.
\begin{proof}
By the local asymptotic analysis of equations $\eqref{stationary Klein-Gordon}$,
locally near $r = r_{+}$ we have \begin{equation*}
\phi =B_{1}\phi_{1}+B_{2}\log(r-r_{+})\phi_{2},\quad r_{+}<r<r_{+}+\epsilon,
\end{equation*}
where $\phi_{1}$ and $\phi_{2}$ are analytic functions on $(r_{+},r_{+}+\epsilon)$ and finite at $r = r_{+}$ and $B_{1}$ and $B_{2}$ are two constants. If we can show that when $B_{2} = 0$, the solution $\phi$ can be extended to the whole domain $(r_{+},\infty)$, then by the asymptotic analysis at $r = \infty$, we have\begin{equation*}
\phi = C_{D}u_{D}(r)+C_{N}u_{N}(r),
\end{equation*}
where $u_{D}$ and $u_{N}$ are the local solutions of $\eqref{stationary Klein-Gordon}$ at $r = \infty$ satisfying $\eqref{Dirichlet asymptotic}$ and $\eqref{Neumann asymptotic}$ respectively.

Note in $\eqref{stationary Klein-Gordon}$, the term $\frac{A_{t}^{2}}{\Omega^{2}}$ is defined everywhere on $[r_{+},\infty)$ since\begin{equation*}
\lim_{r\rightarrow r_{+}}\frac{A_{t}^{2}}{\Omega^{2}}(r) = \frac{e^{2}}{r_{+}^{4}T}\lim_{r\rightarrow r_{+}}\frac{(r-r_{+})^{2}}{r-r_{+}} = 0.
\end{equation*}
Multiplying $\frac{1}{r^{2}}\frac{d\phi}{dr}$ on both sides of the equation $\eqref{stationary Klein-Gordon}$, we have\begin{equation}
\frac{1}{2}\left(\Omega^{2}\left(\frac{d\phi}{dr}\right)^{2}+\frac{1}{l^{2}}\alpha\phi^{2}+\frac{q_{0}^{2}A_{t}^{2}}{\Omega^{2}}\phi^{2}\right) = -\left(\frac{1}{2}\frac{d\Omega^{2}}{dr}+\frac{2}{r}\Omega^{2}\right)\left(\frac{d\phi}{dr}\right)^{2}+\frac{d}{dr}\left(\frac{q_{0}^{2}A_{t}^{2}}{\Omega^{2}}\right)\phi^{2}.\label{linear stimulated energy}
\end{equation}
Integrating $\eqref{linear stimulated energy}$, we have\begin{equation}
\frac{1}{2}\Omega^{2}\left(\frac{d\phi}{dr}\right)^{2}(r)+\frac{1}{l^{2}}\alpha\phi^{2}(r)+\frac{q_{0}^{2}A_{t}^{2}}{2\Omega^{2}}\phi^{2}(r)\leq \frac{1}{2l^{2}}\alpha\phi^{2}(0)+\int_{r_{+}}^{r}\frac{d}{dr}\left(\frac{q_{0}^{2}A_{t}^{2}}{\Omega^{2}}\right)\phi^{2}\mathrm{d}\bar{r}.
\end{equation}
Then by the Gronwall inequality, we have\begin{equation*}
\phi^{2}(r)\leq\phi^{2}(0)e^{\int_{r_{+}}^{r}\frac{d}{dr}\left(\frac{q_{0}^{2}A_{t}^{2}}{\Omega^{2}}\right)\mathrm{d}\bar{r}},
\end{equation*}
which means $\phi(r)$ and $\frac{d\phi}{dr}$ are finite on any interval $[r_{+},R)$. Hence, by the extension principle of the ODE, we know the solution of $\eqref{stationary Klein-Gordon}$ exists on $(r_{+},\infty)$.
\end{proof}

We can also prove the following Wronskian estimate of $\{u_{D},u_{N}\}$.
\begin{proposition}
For the local basis $\{u_{D},u_{N}\}$ of solutions of the linear Klein--Gordon equation $\eqref{stationary Klein-Gordon}$ in Proposition $\ref{trival theorem}$, for $N$ large and $\epsilon$ small enough, we have the following bounds\begin{align}
&\vert u_{D}^{\prime}u_{N}-u_{D}u^{\prime}_{N}\vert\geq Cr^{-4},\quad r>N,\label{infty wronskian}\\&
\vert u_{D}^{\prime}u_{N}-u_{D}u^{\prime}_{N}\vert \geq  C\frac{1}{r-r_{+}}, \quad r_{+}<r<r_{+}+\epsilon.\label{horizon wronskian}
\end{align}
\end{proposition}
\begin{proof}
By the asymptotic behavior of $u_{D}$ and $u_{N}$, $\eqref{infty wronskian}$ follows trivially. Let\begin{equation*}
\{\phi_{1},\log(r-r_{+})\phi_{2}\}
\end{equation*}
be the local basis of the solution of $\eqref{stationary Klein-Gordon}$ near $r = r_{+}$. We have\begin{align*}
&u_{D} = A_{D}\phi_{1}+B_{D}\log(r-r_{+})\phi_{2},\\&
u_{N} = A_{N}\phi_{1}+B_{N}\log(r-r_{+})\phi_{2}.
\end{align*}
Calculating the Wronskian, we have\begin{equation}
\begin{aligned}
&\vert u_{D}^{\prime}u_{N}-u_{D}u^{\prime}_{N}\vert\\=&
\left(A_{D}\phi_{1}^{\prime}+B_{D}\log(r-r_{+})\phi_{2}^{\prime}+B_{D}\frac{\phi_{2}}{r-r_{+}}\right)\left(A_{N}\phi_{1}+B_{N}\log(r-r_{+})\phi_{2}\right)\\&-
\left(A_{D}\phi_{1}+B_{D}\log(r-r_{+})\phi_{2}\right)\left(A_{N}\phi_{1}^{\prime}+B_{N}\log(r-r_{+})\phi_{2}^{\prime}+B_{N}\frac{\phi_{2}}{r-r_{+}}\right)
\\\geq& \frac{1}{2}\frac{\left\vert\phi_{1}(r_{+})\phi_{2}(r_{+})\right\vert}{r-r_{+}}\left\vert A_{D}B_{N}-A_{N}B_{D}\right\vert.
\end{aligned}
\end{equation}
Since $u_{D}$ and $u_{N}$ are linearly independent, then we have\begin{equation*}
\vert A_{D}B_{N}-A_{N}B_{D}\vert>0.
\end{equation*}
Hence, we can prove $\eqref{horizon wronskian}$.
\end{proof}
\section{Proof of Theorem $\ref{linear hair theorem}$ for Dirichlet boundary conditions}
\label{proof of linear hair theorem}
In this section, we will prove Theorem $\ref{linear hair theorem}$ under Dirichlet boundary conditions. We use the notation $\Omega_{M}^{2}$ to emphasize the role of the parameter $M$ in the following argument.
\subsection{Energy functional and the twisted energy functional}
\label{sec:energy functional}
Recall that we can write the stationary spherically symmetric Klein--Gordon equation $\eqref{Klein-Gordon}$ as\begin{equation}
-\frac{d}{dr}(r^{2}\Omega_{M}^{2}\frac{d\phi}{dr})+\frac{1}{l^{2}}\alpha r^{2}\phi-\frac{q_{0}^{2}A_{t}^{2}}{\Omega_{M}^{2}}r^{2}\phi = 0,\label{rkeq}
\end{equation}
where $A_{t}$ takes the form of\begin{equation*}
A_{t} = -e\left(\frac{1}{r_{+}}-\frac{1}{r}\right).
\end{equation*}
The corresponding energy functional is:\begin{equation}
\label{original energy functional}
L_{M}[f] = \int_{r_{+}}^{\infty}r^{2}\Omega_{M}^{2}\left(\frac{df}{dr}\right)^{2}+\left(\frac{1}{l^{2}}\alpha-\frac{q_{0}^{2}A_{t}^{2}}{\Omega_{M}^{2}}\right) r^{2}f^{2}\mathrm{d}r.
\end{equation}
We call the term\begin{equation*}
V_{M}(r): = \frac{1}{l^{2}}\alpha-\frac{q_{0}^{2}A_{t}^{2}}{\Omega_{M}^{2}}
\end{equation*}
defined in the above energy functional the potential term. One can see $V_{M}(r)$ is always negative on $(r_{+},\infty)$. To overcome the difficulties mentioned in Section \ref{outline of the proof}, we introduce the twisted derivative. Let $\nabla_{r}^{h}$ be the twisted derivative\begin{equation*}
\nabla_{r}^{h}f = h\frac{d}{dr}(h^{-1}f),
\end{equation*}
where the function $h$ is called the twist function. The dual operator of this twisted derivative operator is $h^{-1}\frac{d}{dr}(hf)$. We can rewrite the equation $\eqref{rkeq}$ with respect to the twisted derivative as\begin{equation}
-h^{-1}\frac{d}{dr}\left(r^{2}\Omega_{M}^{2}h^{2}\frac{dh^{-1}\phi}{dr}\right)+\left(V_{M,h}(r)-\frac{q_{0}^{2}A_{t}^{2}}{\Omega_{M}^{2}}\right)r^{2}\phi = 0,
\end{equation}
where the twisted potential $V_{M,h}(r)$ is\begin{equation}
V_{M,h}(r):=\frac{1}{l^{2}}\alpha-\frac{1}{r^{2}h}\frac{d}{dr}\left(r^{2}\Omega_{M}^{2}\frac{dh}{dr}\right)
\end{equation}
We can also define the twisted energy functional by\begin{equation}
\label{twisted energy functional}
L_{M,h}[f] := \int_{r_{+}}^{\infty}r^{2}\Omega_{M}^{2}h^{2}\left(\frac{dh^{-1}f}{dr}\right)^{2}+\left(V_{M,h}(r)-\frac{q_{0}^{2}A_{t}^{2}}{\Omega_{M}^{2}}\right)r^{2}f^{2}dr.
\end{equation}
In the proof of Theorem $\ref{linear hair theorem}$ with Dirichlet boundary conditions, we choose $h = r^{-\beta}$. Let $V_{M,\beta}$ denote the potential function and $L_{M,\beta}$ denote the twisted energy functional when $h = r^{-\beta}$, then we have
\begin{equation}
V_{M,\beta}(r) =\beta r^{-1}\frac{d\Omega_{M}^{2}}{dr}+\beta(1-\beta)r^{-2}\Omega_{M}^{2}+\frac{1}{l^{2}}\alpha.
\label{potential}
\end{equation}

We can prove that the twisted energy functional is equivalent to the original energy functional.
\begin{lemma}
\label{change}
If $f\in C^{\infty}_{c}(r_{+},\infty)$, then $L_{M}[f] = L_{M,h}[f]$ for any smooth function $h$.
\end{lemma}
\begin{proof}
This can be proved by direct computation. We have\begin{align*}
&L_{M,h}[f] \\=& \int_{r_{+}}^{\infty}r^{2}h^{2}\Omega_{M}^{2}\left(\frac{dh^{-1}f}{dr}\right)^{2}+\left(V_{M,h}-\frac{q_{0}^{2}A_{t}^{2}}{\Omega_{M}^{2}}\right)r^{2}f^{2}dr\\=&
\int_{r_{+}}^{\infty} r^{2}\Omega_{M}^{2}\left(\frac{df}{dr}\right)^{2}+r^{2}\Omega_{M}^{2}h^{-2}\left(\frac{dh}{dr}\right)^{2}f^{2}-r^{2}\Omega_{M}^{2}h^{-1}\frac{dh}{dr}\frac{df^{2}}{dr}+\left(V_{M,h}(r)-\frac{q_{0}^{2}A_{t}^{2}}{\Omega_{M}^{2}}\right)r^{2}f^{2}dr\\=&
\int_{r_{+}}^{\infty}r^{2}\Omega_{M}^{2}\left(\frac{df}{dr}\right)^{2}+\left(V_{M,h}(r)-\frac{q_{0}^{2}A_{t}^{2}}{\Omega_{M}^{2}}+\Omega_{M}^{2}h^{-2}\left(\frac{dh}{dr}\right)^{2}+\frac{1}{r^{2}}\frac{d}{dr}\left(r^{2}\Omega_{M}^{2}h^{-1}\frac{dh}{dr}\right)\right)r^{2}f^{2}dr\\=&
L_{M}[f].
\end{align*}
The third identity is due to the integration by parts and the fact that $f$ is compactly supported.
\end{proof}
Since $\Omega_{M}^{2}$ degenerates at the event horizon, we define the perturbed energy functional $L^{\epsilon}_{M}[f]$ to be\begin{equation}
L^{\epsilon}_{M}[f] = \int_{r_{+}+\epsilon}^{\infty}r^{2}\Omega_{M}^{2}(\frac{df}{dr})^{2}+\left(\frac{1}{l^{2}}\alpha-\frac{q_{0}^{2}A_{t}^{2}}{\Omega^{2}}\right)r^{2}f^{2}\mathrm{d}r.
\label{perturbated energy functional}
\end{equation}
Similarly, we can define the perturbed twisted energy functional\begin{equation}
L_{M,\beta}^{\epsilon}[f] = \int_{r_{+}+\epsilon}^{\infty}r^{-2\beta+2}\Omega_{M}^{2}\left(\frac{dr^{\beta}f}{dr}\right)^{2}+r^{2}\left(V_{M,\beta}(r)-\frac{q_{0}^{2}A_{t}^{2}}{\Omega_{M}^{2}}\right)f^{2}\mathrm{d}r.
\end{equation}

\subsection{Negative energy bound state}
\label{sec:negative}
Recall the surface gravity of the event horizon\begin{equation}
\kappa := \frac{d\Omega_{M}^{2}}{dr}(r_{+}) = \frac{2}{r_{+}^{2}}(M_{0}-M).
\end{equation}
For fixed parameters $(r_{+},l,\alpha)$, the surface gravity $\kappa$ is determined by the black hole mass $M$. 

We say $L_{M}[f]$ has a negative energy bound state if there exists $f\in C^{\infty}_{c}(r_{+},\infty)$ such that $L_{M}[f]<0$. Let 
\begin{equation*}
\mathcal{A}_{s} = \{M_{e = 0}<M<s,\ \exists f\in C_{c}^{\infty}(r_{+},\infty)\text{ such that } L_{M}[f]<0 \}
\end{equation*}
be the set of all admissible $M$ such that $L_{M}$ admits a negative energy bound state. Due to Lemma $\ref{change}$, if $L_{M}[f]$ has a negative energy bound state, so does $L_{M,h}[f]$ for any smooth function $h$. If the charge $\vert q_{0}\vert$ can be taken to be large, then the negative energy bound state follows trivially. We can prove the following lemma.
\begin{lemma}
\label{charged negative energy bound state}
For any fixed sub-extremal parameters $(M_{b},r_{+},l,\alpha)$ satisfying the bound $-\frac{9}{4}<\alpha<0$, there exists a $q_{1}(M_{b},r_{+},l,\alpha)>0$ such that for any $\vert q_{0}\vert>q_{1}$, we can find a negative energy bound state for the functional $L_{M_{b}}[f]$.
\end{lemma}
\begin{proof}
Let $\eta\equiv1$ on $(r_{+}+\frac{1}{4},r_{+}+\frac{3}{4})$ be a smooth function defined on $\mathbb{R}$ with support on $(r_{+},r_{+}+1)$. Let\begin{align*}
&a_{1}(M_{b},r_{+},l,\alpha) : = \int_{r_{+}}^{\infty}r^{2}\Omega_{M_{b}}^{2}\left(\frac{d\eta}{dr}\right)^{2}+\frac{1}{l^{2}}\alpha r^{2}\eta^{2}\mathrm{d}r,\\&
a_{2}(M_{b},r_{+},l,\alpha): = \int_{r_{+}}^{\infty}\frac{A_{t}^{2}}{\Omega_{M_{b}}^{2}}r^{2}\eta^{2}\mathrm{d}r>0.
\end{align*} 
If we set $q_{1}^{2} = \frac{a_{1}}{a_{2}}$, then for any $q_{0}^{2}>q_{1}^{2}$, we have a negative energy bound state.
\end{proof}
For the second case in Theorem \ref{linear hair theorem} where the sufficient largeness of $q_{0}$ is missing, the existence of a negative energy bound state becomes more subtle. The numerical method in physics literature \cite{hartnoll2008holographic} fails to find a negative energy bound state for $L_{M}[f]$ when $M_{e = 0}<M<M_{0}$ and $q_{0} = 0$. Next, by using a sharp Hardy-type inequality and continuity argument, we prove the existence of a negative energy bound state for the general fixed charge case, which paves the way to prove Theorem $\ref{linear hair theorem}$.
\begin{lemma}\textup{[Sharp Hardy-type Inequality]}
\label{hardy}
Assume  $f\in C_{c}^{\infty}(r_{+},r_{+}+1)$, then we have\begin{equation*}
\int_{r_{+}}^{r_{+}+1}f^{2}(r)\mathrm{d}r\leq 4\int_{r_{+}}^{r_{+}+1}(r-r_{+})^{2}(f^{\prime})^{2}\mathrm{d}r.
\end{equation*}
The constant $4$ in the inequality is sharp in the sense that for any $\delta>0$ small, we can find a $f_{\delta}\in C_{c}^{\infty}(r_{+},r_{+}+1)$ such that \begin{equation*}
\int_{r_{+}}^{r_{+}+1}f_{\delta}^{2}\mathrm{d}r\geq (4-\delta)\int_{r_{+}}^{r_{+}+1}(r-r_{+})^{2}(f_{\delta}^{\prime})^{2}\mathrm{d}r.
\end{equation*}
\end{lemma}
\begin{proof}
By the integration by parts and Hölder inequality, we have\begin{align*}
\int_{r_{+}}^{r_{+}+1}f^{2}\mathrm{d}r &= (r-r_{+})f^{2}(r)\Bigl|_{r_{+}}^{r_{+}+1}-\int_{r_{+}}^{r_{+}+1}2(r-r_{+})f(r)f^{\prime}(r)\mathrm{d}r\\&
\leq 2\left(\int_{r_{+}}^{r_{+}+1}f^{2}\mathrm{d}r\right)^{\frac{1}{2}}\left(\int_{r_{+}}^{r_{+}+1}(r-r_{+})^{2}(f^{\prime})^{2}\mathrm{d}r\right)^{\frac{1}{2}}.
\end{align*}
Thus we have\begin{equation*}
\int_{r_{+}}^{r_{+}+1}f^{2}\mathrm{d}r\leq 4\int_{r_{+}}^{r_{+}+1}(r-r_{+})^{2}\left(\frac{df}{dr}\right)^{2}\mathrm{d}r.
\end{equation*}
Next, we prove that the constant $4$ here is sharp. For $\epsilon>0$ and $-\frac{1}{2}<\alpha_{1}<\alpha_{2}<0$, let $f_{\epsilon,\alpha_{1},\alpha_{2}}$ be\begin{equation}
f_{\epsilon,\alpha_{1},\alpha_{2}} = \left\{
\begin{aligned}
&A_{1}(r-r_{+})^{\alpha_{1}}+A_{2}(r-r_{+})^{\alpha_{2}}-(A_{1}+A_{2}),\quad r-r_{+}>\epsilon,\\
&0,\quad 0<r-r_{+}\leq\epsilon,
\end{aligned}
\right.
\end{equation}
where $A_{1} = \epsilon^{\alpha_{2}-\alpha_{1}}-\epsilon^{-\alpha_{1}}$ and $A_{2} = -(1-\epsilon^{-\alpha_{1}})$. We have\begin{align*}
&\int_{r_{+}}^{r_{+}+1}f_{\epsilon,\alpha_{1},\alpha_{2}}^{2}\mathrm{d}r=   \int_{r_{+}+\epsilon}^{r_{+}+1}f_{\epsilon,\alpha_{1},\alpha_{2}}^{2}\mathrm{d}r \\=&
A_{1}^{2}\left(\frac{1-\epsilon^{2\alpha_{1}+1}}{2\alpha_{1}+1}-2\frac{1-\epsilon^{\alpha_{1}+1}}{\alpha_{1}+1}+1-\epsilon\right)+A_{2}^{2}\left(\frac{1-\epsilon^{2\alpha_{2}+1}}{2\alpha_{2}+1}-2\frac{1-\epsilon^{\alpha_{2}+1}}{\alpha_{2}+1}+1-\epsilon\right)\\&+
2A_{1}A_{2}\left(\frac{1-\epsilon^{\alpha_{1}+\alpha_{2}+1}}{\alpha_{1}+\alpha_{2}+1}-2\frac{1-\epsilon^{\alpha_{1}+1}}{\alpha_{1}+1}-2\frac{1-\epsilon^{\alpha_{2}+1}}{\alpha_{2}+1}+1-\epsilon\right),
\end{align*}
and \begin{align*}
&\int_{r_{+}}^{r_{+}+1}(r-r_{+})^{2}\bigl(\frac{df_{\epsilon,\alpha_{1},\alpha_{2}}}{dr}\bigr)^{2}\mathrm{d}r \\=& 
A_{1}^{2}\alpha_{1}^{2}\frac{1-\epsilon^{2\alpha_{1}+1}}{2\alpha_{1}+1}+A_{2}^{2}\alpha_{2}^{2}\frac{1-\epsilon^{2\alpha_{2}+1}}{2\alpha_{2}+1}+2A_{1}A_{2}\alpha_{1}\alpha_{2}\frac{1-\epsilon^{\alpha_{1}+\alpha_{2}+1}}{\alpha_{1}+\alpha_{2}+1}.
\end{align*}
Hence if we choose $\epsilon$ small and $\alpha_{1},\alpha_{2}$ close to $-\frac{1}{2}$, we have \begin{equation*}
\int_{r_{+}}^{r_{+}+1}f_{\epsilon,\alpha_{1},\alpha_{2}}^{2}\mathrm{d}r\geq (4-2\delta)\int_{r_{+}}^{r_{+}+1}(r-r_{+})^{2}\left(\frac{df_{\epsilon,\alpha_{1},\alpha_{2}}}{dr}\right)^{2}\mathrm{d}r.
\end{equation*}
Note by the construction of $f_{\epsilon,\alpha_{1},\alpha_{2}}$, we know that $f_{\epsilon,\alpha_{1},\alpha_{2}}$ is continuous on $[r_{+},r_{+}+1]$ and smooth on $(r_{+}+\epsilon,r_{+}+1)$. Hence $f_{\epsilon,\alpha_{1},\alpha_{2}}\in H_{0}^{1}(r_{+}+\epsilon,r_{+}+1)$. Then we can approximate $f_{\epsilon,\alpha_{1},\alpha_{2}}$ by $f_{\delta}\in C_{c}(r_{+}+\epsilon,r_{+}+1)$, we have $f_{\delta}\in C_{c}^{\infty}(r_{+},r_{+}+1)$ such that\begin{equation*}
\int_{r_{+}}^{r_{+}+1}(f_{\delta})^{2}\mathrm{d}r\geq (4-\delta)\int_{r_{+}}^{r_{+}+1}(r-r_{+})^{2}\left(\frac{df_{\delta}}{dr}\right)^{2}\mathrm{d}r.
\end{equation*}
This concludes the proof of the lemma.
\end{proof}
\begin{remark}
Note that the inequality we got above is scaling invariant. So we can get the same result with the same sharp constant if we change the domain in the setting to be $[r_{+},r_{+}+a]$ for any $a>0$. In the later discussion, $a$ will be chosen to be a small number.
\end{remark}
Now we are ready to prove the following lemma showing that for fixed $q_{0}$ and the parameters $(r_{+},l,\alpha,q_{0})$ satisfying conditions \eqref{hardy assumption} and \eqref{large condition}, the set $\mathcal{A}_{M_{0}}$ is non-empty.
\begin{lemma}
\label{nonempty}
For each fixed parameters$ (r_{+},l,\alpha,q_{0})$ satisfying\begin{align}
&-\frac{9}{4}<\alpha<\min\left\{0,-\frac{3}{2}+\frac{1}{2}l^{2}q_{0}^{2}\right\},\label{revised condition1}\\&
\frac{r_{+}^{2}}{l^{2}}>R_{0},\label{revised condition2}
\end{align}
where $R_{0}$ is the positive root of \eqref{quadratic}, $\mathcal{A}_{M_{0}}$ is non-empty.
\end{lemma}
\begin{proof}
For fixed $(r_{+},l,\alpha,q_{0})$, if we can find the negative energy bound state for $L_{M_{0}}[f]$, then by continuity of $L_{M}[f]$ with respect to $M$ for fixed $f\in C_{c}(r_{+},\infty)$, we can find the negative energy bound state for $L_{M}[f]$ with $M$ in the neighborhood of $M_{0}$.

Calculating $L_{M}[f]$ with $M = M_{0}$, we have\begin{align*}
L_{M_{0}}[f] &= \int_{r_{+}}^{\infty}\frac{1}{l^{2}}(r-r_{+})^{2}\left(r^{2}+2r_{+}r+3r_{+}^{2}+l^{2}\right)\left(\frac{df}{dr}\right)^{2}+\left(\frac{1}{l^{2}}\alpha-\frac{q_{0}^{2}A_{t}^{2}}{\Omega^{2}}\right)r^{2}f^{2}\mathrm{d}r.
\end{align*}
By Lemma $\ref{hardy}$, we can find $f_{\delta}\in C_{c}^{\infty}[r_{+},\infty)$ with $f_{\delta}(r) = 0$ for any $r>r_{+}+\epsilon$ such that\begin{equation*}
\int_{r_{+}}^{r_{+}+\epsilon}f_{\delta}^{2}\mathrm{d}r\geq (4-\delta)\int_{r_{+}}^{r_{+}+\epsilon}(r-r_{+})^{2}\left(\frac{df_{\delta}}{dr}\right)^{2}\mathrm{d}r.
\end{equation*}
Then we have the following estimate for $L_{M_{0}}[f_{\delta}]$:
\begin{align*}
&L_{M}[f_{\delta}] \\<&\frac{1}{l^{2}}\int_{r_{+}}^{r_{+}+\epsilon}(r-r_{+})^{2}\left(r^{2}+2r_{+}r+3r_{+}^{2}+l^{2}\right)\left(\frac{df_{\delta}}{dr}\right)^{2}\mathrm{d}r\\&+\frac{1}{l^{2}}(4-\delta)r_{+}^{2}\int_{r_{+}}^{r_{+}+\epsilon}\left(\alpha-l^{2}q_{0}^{2}\frac{l^{2}+3r_{+}^{2}}{l^{2}+6r_{+}^{2}}\right)(r-r_{+})^{2}\left(\frac{df_{\delta}}{dr}\right)^{2}\mathrm{d}r\\
<&\frac{1}{l^{2}}r_{+}^{2}\int_{r_{+}}^{r_{+}+\epsilon}(r-r_{+})^{2}\left(6+\frac{l^{2}}{r_{+}^{2}}-4\left(\alpha-\frac{1}{2}l^{2}q_{0}^{2}-\frac{l^{2}q_{0}^{2}}{2}\frac{l^{2}}{l^{2}+6r_{+}^{2}}\right)+C\epsilon+C\delta\right)\left(\frac{df_{\delta}}{dr}\right)^{2}\mathrm{d}r\\=&
\frac{1}{l^{2}}r_{+}^{2}\int_{r_{+}}^{r_{+}+\epsilon}(r-r_{+})^{2}\left(-4\left(\alpha-\frac{3}{2}-\frac{l^{2}q_{0}^{2}}{2}\right)+\frac{l^{2}}{r_{+}^{2}}+\frac{2l^{4}q_{0}^{2}}{l^{2}+6r_{+}^{2}}+C\epsilon+C\delta\right)\left(\frac{df_{\delta}}{dr}\right)^{2}\mathrm{d}r.
\end{align*}
Then, by the conditions \eqref{revised condition1} and \eqref{revised condition2}, we can choose $\epsilon$ and $\delta$ small such that the right-hand side of the above inequality is negative.

Now we have $L_{M_{0}}[f_{\delta}]<0$. By the continuity, we can prove $\mathcal{A}_{M_{0}}$ is non-empty.
\end{proof}

\subsection{Minimizer of the energy functional $L_{M}^{\epsilon}[f]$}
We consider the following function class $\mathcal{F}$:\begin{equation*}
\mathcal{F}: = \left\{f\in C^{\infty}_{c}(r_{+},\infty),\quad \int_{r_{+}}^{\infty}\frac{r^{2}f^{2}}{\Omega^{2}}\mathrm{d}r = 1\right\}.
\end{equation*}
We can define the perturbed function class $\mathcal{F}^{\epsilon}$ by changing the domain in the above definition to be $(r_{+}+\epsilon,\infty)$. Since $f\in\mathcal{F}^{\epsilon}$ is compactly supported, we still have $L_{M}^{\epsilon}[f] = L_{M,\beta}^{\epsilon}[f]$, similar to Lemma $\ref{change}$. 

At first glimpse one may wonder whether the minimum of $L_{M}^{\epsilon}[f]$ can be attained in $\mathcal{F}^{\epsilon}$ since $L_{M}^{\epsilon}$ contains the integral of $r^{2}f^{2}$ while $\mathcal{F}^{\epsilon}$ only makes the restriction on the $L^{2}$ norm of the function near infinity. However, we can prove the following proposition.\begin{proposition}
For $M_{e=0}<M<M_{0}$, if $L_{M}[f]$ has a negative energy bound state, then for any $\epsilon>0$ small enough, $L_{M}^{\epsilon}[f]$ can attain its negative minimum in the function class $\mathcal{F}^{\epsilon}$.
\label{eigenvalue}
\end{proposition}
\begin{proof}
In this proof, we use $C$ to denote a constant independent of $\epsilon$. By Lemma $\ref{change}$ and the continuity of $L_{M,\frac{3}{2}}^{\epsilon}[f]$ with respect to $\epsilon$, we know that for $\epsilon$ small enough, we can also find a negative energy bound state for $L_{M,\frac{3}{2}}^{\epsilon}[f]$. It remains to prove that $L_{M,\frac{3}{2}}^{\epsilon}$ can attain its minimum in $\mathcal{F}^{\epsilon}$. We have\begin{align*}
L_{M,\frac{3}{2}}^{\epsilon}[f] &= \int_{r_{+}+\epsilon}^{\infty}r^{-1}\Omega_{M}^{2}\left(\frac{dr^{\frac{3}{2}}f}{dr}\right)^{2}+\left(V_{M,\frac{3}{2}}(r)-\frac{q_{0}^{2}A_{t}^{2}}{\Omega_{M}^{2}}\right)r^{2}f^{2}\mathrm{d}r,\\
V_{M,\frac{3}{2}}(r) &= \frac{1}{l^{2}}\left(\frac{9}{4}+\alpha\right)-\frac{3}{4r^{2}}+\frac{9M}{2r^{3}}-\frac{15e^{2}}{4r^{4}}.
\end{align*}
Since $L_{M}[f]$ has a negative energy bound state, the term \begin{equation*}
V_{M,\frac{3}{2}}(r)-\frac{q_{0}^{2}A_{t}^{2}}{\Omega_{M}^{2}}
\end{equation*}
must have negative values. However, by considering the limit of this term when $r\rightarrow\infty$, we have\begin{equation*}
\lim_{r\rightarrow\infty}V_{M,\frac{3}{2}}(r)-\frac{q_{0}^{2}A_{t}^{2}}{\Omega_{M}^{2}} = \frac{1}{l^{2}}\left(\frac{9}{4}+\alpha\right)>0.
\end{equation*}
Hence there exists a $x_{v}$ such that $V_{M,\frac{3}{2}}-\frac{q_{0}^{2}A_{t}^{2}}{\Omega_{M}^{2}}$ is positive for $r>x_{v}$.

Then we have\begin{equation}
\begin{aligned}
L_{M,\frac{3}{2}}^{\epsilon}[f]-&\int_{r_{+}+\epsilon}^{x_{v}}\left(V_{M,\frac{3}{2}}(r)-\frac{q_{0}^{2}A_{t}^{2}}{\Omega_{M}^{2}}\right)r^{2}f^{2}\mathrm{d}r \\&= \int_{r_{+}+\epsilon}^{\infty}r^{-1}\Omega_{M}^{2}\left(\frac{dr^{\frac{3}{2}}f}{dr}\right)^{2}\mathrm{d}r+\int_{x_{v}}^{\infty}\left(V_{M,\frac{3}{2}}(r)-\frac{q_{0}^{2}A_{t}^{2}}{\Omega_{M}^{2}}\right)r^{2}f^{2}\mathrm{d}r
\end{aligned}
\label{core equality1}
\end{equation}
Since on $(r_{+},x_{v})$\begin{equation*}
\left\vert V_{M,\frac{3}{2}}(r)-\frac{q_{0}^{2}A_{t}^{2}}{\Omega_{M}^{2}}\right\vert\leq C,
\end{equation*}
we have\begin{equation*}
L_{M,\frac{3}{2}}^{\epsilon}[f]-\int_{r_{+}+\epsilon}^{x_{v}}\left(V_{M,\frac{3}{2}}(r)-\frac{q_{0}^{2}A_{t}^{2}}{\Omega_{M}^{2}}\right)r^{2}f^{2}\mathrm{d}r\leq L_{M,\frac{3}{2}}^{\epsilon}[f]+C\int_{r_{+}+\epsilon}^{x_{v}}r^{2}f^{2}\mathrm{d}r.
\end{equation*}

Since $f\in\mathcal{F}^{\epsilon}$, we have \begin{equation}
\int_{r_{+}+\epsilon}^{x_{v}}r^{2}f^{2}\mathrm{d}r = \int_{r_{+}+\epsilon}^{\infty}\frac{r^{2}}{\Omega_{M}^{2}}\Omega_{M}^{2}f^{2}\mathrm{d}r\leq C.
\end{equation}
Since the right hand side of $\eqref{core equality1}$ is positive, we have\begin{equation*}
L_{M,\frac{3}{2}}^{\epsilon}[f]>-C.
\end{equation*}
Hence we have $L_{M,\frac{3}{2}}^{\epsilon}$ is lower bounded. Let $f_{n}^{\epsilon}\in\mathcal{F}^{\epsilon}$ be the minimizing sequence of $L_{M,\frac{3}{2}}^{\epsilon}$. Without loss of generality, we can assume $L_{M,\frac{3}{2}}^{\epsilon}[f_{n}^{\epsilon}]<0$. Then since $L_{M,\frac{3}{2}}^{\epsilon}[f_{n}^{\epsilon}]$ is bounded, by $\eqref{core equality1}$ we know $rf_{n}^{\epsilon}$ is $L^{2}$ integrable:\begin{equation}
\int_{r_{+}+\epsilon}^{\infty}r^{2}(f_{n}^{\epsilon})^{2}\mathrm{d}r\leq C,
\label{increase the integrability}
\end{equation}
For $r>r_{+}+\epsilon$, we have\begin{equation*}
\Omega_{M}^{2}(r)>C\epsilon r^{2}.
\end{equation*}
Hence we have\begin{align*}
\int_{r_{+}+\epsilon}^{\infty}r\left(\frac{dr^{\frac{3}{2}}f_{n}^{\epsilon}}{dr}\right)^{2}\mathrm{d}r\leq C(\epsilon).
\end{align*}
Then we have that $f_{n}^{\epsilon}$ is $H^{1}$ bounded. By Rellich compactness theorem, we have $f_{n}^{\epsilon}$ weakly converges to $\phi^{\epsilon}$ in $H_{0}^{1}$ and strongly converges to $\phi^{\epsilon}$ in $L^{2}$ on any compact set $K\subset[r_{+}+\epsilon,\infty)$. Next, we prove that $\phi^{\epsilon}$ also belongs to $\mathcal{F}^{\epsilon}$. By the strong $L^{2}$ convergence on a compact set $K$, we have\begin{equation*}
\int_{K} \frac{r^{2}}{\Omega_{M}^{2}}(\phi^{\epsilon})^{2}\mathrm{d}r\leq1.
\end{equation*}
Passing to the limit, we have \begin{equation*}
\int_{r_{+}+\epsilon}^{\infty}\frac{r^{2}}{\Omega_{M}^{2}}\left(\phi^{\epsilon}\right)^{2}\mathrm{d}r\leq 1.
\end{equation*}
Assume $\int_{r_{+}+\epsilon}^{\infty}\frac{r^{2}}{\Omega_{M}^{2}}(\phi^{\epsilon})^{2}\mathrm{d}r<1$, then for any $N$, there exist infinite many $f_{n}^{\epsilon}$ such that \begin{equation*}
\int_{N}^{\infty}\frac{r^{2}}{\Omega_{M}^{2}}(f_{n}^{\epsilon})^{2}\mathrm{d}r>\frac{1}{2}\left(1-\int_{r_{+}+\epsilon}^{\infty}\frac{r^{2}}{\Omega_{M}^{2}}(\phi^{\epsilon})^{2}\mathrm{d}r\right).
\end{equation*}
Then by $\eqref{core equality1}$, we have \begin{equation*}
C\geq \int_{x_{v}}^{\infty}\left(V_{M,\frac{3}{2}}-\frac{q_{0}^{2}A_{t}^{2}}{\Omega_{M}^{2}}\right)r^{2}(f_{n}^{\epsilon})^{2}\mathrm{d}r\geq cN^{2}\int_{N}^{\infty}(f_{n}^{\epsilon})^{2}\mathrm{d}r,
\end{equation*}
which provides a contradiction if $N$ is large enough. Hence $L_{M,\frac{3}{2}}^{\epsilon}[f]$ can attain its minimum $L_{M,\frac{3}{2}}^{\epsilon}[\phi^{\epsilon}]$ in $\mathcal{F}^{\epsilon}$.
\end{proof}
\begin{remark}
\label{coercivity of infinity}
The key point in obtaining the minimizer in the above proof is that the twisted potential $V_{M,\frac{3}{2}}$ is strictly positive near infinity. This property allows us to increase the integrability of the functions in the minimizing sequence and prove that $\phi^{\epsilon}\in\mathcal{F}^{\epsilon}$. For the $q_{0} = 0$ case, we have a more refined description of $V_{M,\frac{3}{2}}$: $V_{M,\frac{3}{2}}$ is negative on $(r_{+},x_{v})$ and positive on $(x_{v},\infty)$. The presence of the charge $q_{0}$ complicates the sign of the term $V_{M,\frac{3}{2}}-\frac{q_{0}^{2}A_{t}^{2}}{\Omega_{M}^{2}}$ on $(r_{+},x_{v})$.
\end{remark}
By the constrained variational principle, we can get the Euler--Lagrange equation for $\phi^{\epsilon}$\begin{equation}
\label{Euler for phi}
-\frac{d}{dr}\left(r^{2}\Omega_{M}^{2}\frac{d\phi^{\epsilon}}{dr}\right)+\left(\frac{1}{l^{2}}\alpha-\frac{q_{0}^{2}A_{t}^{2}}{\Omega_{M}^{2}}\right)r^{2}\phi^{\epsilon} = -\lambda_{M}^{\epsilon}\frac{r^{2}}{\Omega_{M}^{2}}\phi^{\epsilon},
\end{equation}
where $\lambda_{M}^{\epsilon}$ is the minimum of $L_{M,\frac{3}{2}}^{\epsilon}$ in $\mathcal{F}^{\epsilon}$. 

\begin{remark}
One can still have Lemma $\ref{eigenvalue}$ and estimate $\eqref{increase the integrability}$ without assuming the negative energy bound state. The proof follows line by line. However, the negative energy bound state condition allows us to derive the Euler-Lagrange equation with negative eigenvalue $-\lambda_{M}$. 
\end{remark}

Applying the asymptotic analysis to the solution of $\eqref{Euler for phi}$, asymptotically we have $\phi^{\epsilon}\approx r^{-\frac{3}{2}-\Delta}$ since $\Vert \phi^{\epsilon}\Vert_{L^{2}}$ is bounded. Moreover, by using the energy estimate, we can get the following uniform bound for $\phi^{\epsilon}$ independent of $\epsilon$.
\begin{proposition}
\label{gain regularity}
Let $\phi^{\epsilon}$ be the solution of $\eqref{Euler for phi}$ with sub-extremal parameters $(M,r_{+},l,\alpha,q_{0})$ obtained above. Then $\phi^{\epsilon}$ is asymptotically $Cr^{-\frac{3}{2}-\Delta}$ when $r$ approaches infinity. Moreover, we have\begin{equation}
\label{gain}
\int_{r_{+}+\epsilon}^{\infty}r^{2+\Delta}(\phi^{\epsilon})^{2}\mathrm{d}r<C,
\end{equation}
where $C$ is a constant independent of $\epsilon$.
\end{proposition}
\begin{proof}
By the asymptotic analysis of the equation$\eqref{stationary Klein-Gordon}$ at $r = \infty$ and $L^{2}$ boundedness of $r\phi^{\epsilon}$, we know \begin{equation*}
\phi^{\epsilon} \approx Ar^{-\frac{3}{2}-\Delta},\quad r>N.
\end{equation*}
It remains to prove $\eqref{gain}$. Let \begin{equation}
\phi^{\epsilon} = C_{D}^{\epsilon}(r)u_{D}(r)+C_{N}^{\epsilon}(r)u_{N}(r),\label{varitional constant}
\end{equation}
where $u_{D}$ and $u_{N}$ are the local basis of the equation $\eqref{stationary Klein-Gordon}$ with the same sub-extremal parameters $(M,r_{+},l,\alpha,q_{0})$ on $(r_{+},\infty)$, defined in Proposition $\ref{trival theorem}$. 
By the asymptotic behavior of $\phi^{\epsilon}$ near $r = \infty$, we have
\begin{equation*}
\lim_{r\rightarrow\infty}C_{N}^{\epsilon}(r) = 0.
\end{equation*}
Substituting $\eqref{varitional constant}$ into $\eqref{Euler for phi}$, we have\begin{align*}
(C_{D}^{\epsilon})^{\prime}(r)u_{D}(r)+(C_{N}^{\epsilon})^{\prime}(r)u_{N}(r) &= 0,\\
(C_{D}^{\epsilon})^{\prime}(r)u_{D}^{\prime}(r)+(C_{N}^{\epsilon})^{\prime}(r)u_{N}^{\prime}(r) &= \frac{-\lambda_{M}^{\epsilon}\phi^{\epsilon}}{(\Omega_{M}^{2})^{2}}.
\end{align*}
Then we have
\begin{align*}
C_{N}^{\epsilon}(r) = \int_{\infty}^{r}\frac{\lambda_{M}^{\epsilon}\phi^{\epsilon}u_{D}}{(\Omega_{M}^{2})^{2}\left(u_{N}u_{D}^{\prime}-u_{N}^{\prime}u_{D}\right)}\mathrm{d}\bar{r}.
\end{align*}
Using $\eqref{infty wronskian}$ and $\eqref{horizon wronskian}$, we have\begin{align*}
\left\vert\frac{u_{D}}{u_{N}^{\prime}u_{D}-u_{N}u^{\prime}_{D}}\frac{1}{(\Omega_{M}^{2})^{2}}\right\vert\leq Cr^{-\frac{3}{2}-\Delta},\quad r\rightarrow\infty.
\end{align*}

Hence for $r>r_{+}+1$, we have\begin{equation}
\label{Cd gain}
\begin{aligned}
\vert C_{N}^{\epsilon}(r)\vert &\leq C\lambda_{M}^{\epsilon}\int_{r}^{\infty}\left\vert\frac{u_{D}}{u_{N}^{\prime}u_{D}-u_{N}u^{\prime}_{D}}\frac{\phi^{\epsilon}}{(\Omega_{M}^{2})^{2}}\right\vert\mathrm{d}\bar{r}
\\&\leq C\lambda_{M}^{\epsilon}\left(\int_{r}^{\infty}\bar{r}^{2}(\phi^{\epsilon})^{2}\mathrm{d}\bar{r}\right)^{\frac{1}{2}}\left(\int_{r}^{\infty}\frac{1}{\bar{r}^{2}}\left(\frac{u_{D}}{u_{N}^{\prime}u_{D}-u_{N}u^{\prime}_{D}}\right)^{2}\left(\frac{1}{(\Omega_{M}^{2})^{2}}\right)^{2}\mathrm{d}\bar{r}\right)^{\frac{1}{2}}\\&\leq C\lambda_{M}^{\epsilon}r^{-2-\Delta},\quad r>r_{+}+1,
\end{aligned}
\end{equation}
where $C$ here is a constant independent of $\epsilon$. Similarly, for $C_{D}^{\epsilon}$ and $r>r_{+}+1$, we have\begin{equation}
\label{large Cd}
\left\vert C_{D}^{\epsilon}(r)-C_{D}^{\epsilon}(r_{+}+1)\right\vert\leq C\left(\int_{r_{+}+1}^{r}\left(\frac{1}{r}\frac{u_{N}}{u^{\prime}_{N}u_{D}-u_{N}u^{\prime}_{D}}\frac{1}{(\Omega_{M}^{2})^{2}}\right)^{2}\mathrm{d}\bar{r}\right)^{\frac{1}{2}}\leq C.
\end{equation}
Thus we only need to show that $C_{D}^{\epsilon}(r_{+}+1)$ is uniformly bounded in $\epsilon$. Note that \begin{equation*}
C_{D}^{\epsilon}(r)u_{D}(r) = \phi^{\epsilon}-C_{N}^{\epsilon}(r)u_{N}(r).
\end{equation*}
If we can find a sequence $\epsilon_{n}$ such that $\vert C_{D}^{\epsilon_{n}}(r_{+}+1)\vert$ goes to infinity, then by $\eqref{large Cd}$, $\vert C_{D}^{\epsilon_{n}}(r)\vert$ goes to infinity uniformly in $r$ for $r>r_{+}+1$. Hence we have\begin{equation*}
\int_{r_{+}+1}^{\infty}(C_{D}^{\epsilon_{n}})^{2}(r)u^{2}_{D}(r)\mathrm{d}r\rightarrow \infty,
\end{equation*}
which is a contradiction since $\phi^{\epsilon}$ and $C_{N}u_{N}$ are uniformly $L^{2}$ bounded in $\epsilon$. Therefore, we obtain the uniform boundedness of $C_{N}(r)$. Hence, we can prove $\eqref{gain}$.
\end{proof}

Next, we want to get rid of the $\epsilon$. We can prove the following proposition.
\begin{proposition}
\label{nontrival solution}
If $L_{M}[f]$ has a negative energy bound state, then there exists a non-zero solution of the equation\begin{equation}
-\frac{d}{dr}(r^{2}\Omega_{M}^{2}\frac{d\phi}{dr})+\left(\frac{1}{l^{2}}\alpha-\frac{q_{0}^{2}A_{t}^{2}}{\Omega_{M}^{2}}\right)r^{2}\phi = -\frac{\lambda_{M}r^{2}}{\Omega_{M}^{2}}\phi,\quad r_{+}<r<\infty.
\end{equation}
Moreover, $\phi$ satisfies the Dirichlet boundary condition and can be extended continuously to the event horizon.
\end{proposition}
\begin{proof}
Let $K_{n} = [r_{+}+\frac{1}{n},n]$. By $\eqref{core equality1}$, we have\begin{equation*}
\int_{K_{n}}\left(\frac{d\phi^{\epsilon}}{dr}\right)^{2}\mathrm{d}r\leq C(n).
\end{equation*}
Thus we have the uniform boundedness for $\phi^{\epsilon}$ on $K_{n}$\begin{align}
\Vert \phi^{\epsilon}\Vert_{H^{1}(K_{n})}\leq C(n).
\end{align}
Hence $\phi^{\epsilon}$ weakly converges to $\phi$ in $H^{1}_{loc}$ and strongly converges to $\phi$ in $L^{2}_{loc}$ up to a subsequence. By $\eqref{gain}$, we have \begin{equation*}
C\geq \int_{r_{+}+\epsilon}^{\infty}r^{2+\Delta}(\phi^{\epsilon})^{2}\mathrm{d}r\geq N^{\Delta}\int_{N}^{\infty}r^{2}(\phi^{\epsilon})^{2}\mathrm{d}r. 
\end{equation*}
Then we have\begin{equation}
\int_{N}^{\infty}r^{2}(\phi^{\epsilon})^{2}\mathrm{d}r\leq \frac{C}{N^{\Delta}},\label{c1}
\end{equation}
where $N$ is a positive constant which will be chosen very large later.

When the surface gravity $\kappa > -\frac{1}{l^{2}}\alpha\frac{r_{+}}{\beta}$, we have $V_{M,\beta}(r_{+}) > 0$. Now for our fixed parameters $(M,r_{+},l,\alpha,q_{0})$, we can find $\beta_{0}$ large enough, such that $\frac{d\Omega_{M}^{2}}{dr}(r_{+})>-\frac{1}{l^{2}}\alpha\frac{r_{+}}{\beta_{0}}$. Then we have\begin{align*}
V_{M,\beta_{0}}(r_{+})&>0,\\
\lim_{r\rightarrow\infty}V_{M,\beta_{0}}(r) &= \frac{1}{l^{2}}\left(
\beta_{0}(3-\beta_{0})+\alpha\right)<0.
\end{align*}
Basic calculation shows that $V_{M,\beta_{0}}$ will have exactly one zero point $x_{\beta_{0}}$. Hence there exists $\widetilde{x}_{\beta_{0}}>r_{+}$ such that $V_{M,\beta_{0}}-\frac{q_{0}^{2}A_{t}^{2}}{\Omega_{M}^{2}}$ is positive on $(r_{+},\widetilde{x}_{\beta_{0}})$.

Considering the energy functional $L_{M,\beta_{0}}^{\epsilon}[f]$, we have\begin{equation}
\label{core2}
\begin{aligned}
&L^{\epsilon}_{M,\beta_{0}}[\phi^{\epsilon}] +\int_{\widetilde{x}_{\beta_{0}}}^{\infty}(-r^{2})\left(V_{M,\beta_{0}}-\frac{q_{0}^{2}A_{t}^{2}}{\Omega_{M}^{2}}\right)(\phi^{\epsilon})^{2}\mathrm{d}r \\=& \int_{r_{+}+\epsilon}^{\infty}r^{-2\beta_{0}+2}\Omega_{M}^{2}\left(\frac{d}{dr}(r^{\beta_{0}}\phi^{\epsilon})\right)^{2}\mathrm{d}r+\int_{r_{+}+\epsilon}^{\widetilde{x}_{\beta_{0}}}r^{2}\left(V_{M,\beta_{0}}-\frac{q_{0}^{2}A_{t}^{2}}{\Omega_{M}^{2}}\right)(\phi^{\epsilon})^{2}\mathrm{d}r.
\end{aligned}
\end{equation}
Since \begin{equation*}
L_{M,\beta_{0}}[\phi^{\epsilon}] = -\lambda^{\epsilon}_{M},
\end{equation*}
where $\lambda_{M}^{\epsilon}$ is bounded above due to the fact that $L_{M}[f]$ has a negative energy bound state. Since the term $\left(V_{M,\alpha_{0}}-\frac{q_{0}^{2}A_{t}^{2}}{\Omega_{M}^{2}}\right)$ on the left hand side of $\eqref{core2}$ is bounded, we have \begin{equation}
\int_{\widetilde{x}_{\alpha_{0}}}^{\infty}r^{2}(\phi^{\epsilon})^{2}\mathrm{d}r\geq C_{1}.\label{c2}
\end{equation}
Combining $\eqref{c1}$ and $\eqref{c2}$, we have \begin{equation}
\int_{x_{\beta_{0}}}^{N}r^{2}(\phi^{\epsilon})^{2}\mathrm{d}r\geq C_{1}-\frac{C}{N^{\Delta}}>\frac{C_{1}}{2}
\label{coervicity}
\end{equation}
by choosing $N$ large enough.

Thus $\phi$ is non-zero. By the local convergence result, $\phi$ satisfies the equation \begin{equation}
-\frac{d}{dr}(r^{2}\Omega_{M}^{2}\frac{d\phi}{dr})+\left(\frac{1}{l^{2}}\alpha-\frac{q_{0}^{2}A_{t}^{2}}{\Omega_{M}^{2}}\right)r^{2}\phi = -\lambda_{M}\frac{r^{2}}{\Omega_{M}^{2}}\phi.\label{a random equation}
\end{equation}
By the asymptotic analysis of the equation $\eqref{a random equation}$, we have\begin{align*}
&\phi(r) \approx B_{1}(r-r_{+})^{\frac{1}{\kappa}\sqrt{\lambda_{M}}}+B_{2},\quad r\rightarrow r_{+},\\&
\phi \approx C_{D}r^{-\frac{3}{2}-\Delta}+C_{N}r^{-\frac{3}{2}+\Delta},\quad r\rightarrow\infty.
\end{align*}
Since $\phi^{\epsilon}$ is locally convergent to $\phi$ on any compact set $K\subset [r_{+},\infty)$, we have\begin{align*}
&\int_{K}\frac{r^{2}}{\Omega_{M}^{2}}\phi^{2}<C,\\
&\int_{K}r^{2}\Omega_{M}^{2}(\frac{d\phi}{dr})^{2}+\left(\frac{1}{l^{2}}\alpha-\frac{q_{0}^{2}A_{t}^{2}}{\Omega_{M}^{2}}\right)r^{2}\phi^{2}\mathrm{d}r<C,
\end{align*}
where $C$ is a constant independent of $K$. Thus we conclude that $B_{2} = C_{N} = 0$ and $\phi$ is the desired solution of $\eqref{Euler for phi}$.
\end{proof}
The immediate consequence of Proposition \ref{nontrival solution} is that, $e^{\sqrt{\lambda_{M}}t}\phi$ is a growing mode solution to the uncharged Klein--Gordon equation \eqref{Klein-Gordon} with the mode $\omega = -i\sqrt{\lambda}$.
\begin{remark}
By the asymptotic behavior of $\phi$, we have \begin{equation*}
\int_{r_{+}}^{\infty}\frac{r^{2}}{\Omega_{M}^{2}}\phi^{2}\mathrm{d}r <\infty.
\end{equation*}
If we normalize $\phi$ such that $$\int_{r_{+}}^{\infty}\frac{r^{2}}{\Omega_{M}^{2}}\phi^{2}\mathrm{d}r= 1,$$then we have\begin{equation}
\int_{r_{+}}^{\infty}r^{2}\Omega_{M}^{2}\left(\frac{d\phi}{dr}\right)^{2}+\left(\frac{1}{l^{2}}\alpha-\frac{q_{0}^{2}A_{t}^{2}}{\Omega_{M}^{2}}\right)r^{2}\phi^{2}\mathrm{d}r = -\lambda_{M}.
\end{equation}
\label{normremark}
\end{remark}
\begin{remark}
\label{remark}
Note that in the above argument, to show $\phi$ is non-zero, we need the fact that $\eqref{c2}$ is bounded away from $0$ to have the lower bound $\eqref{coervicity}$. This step relies on the existence of a negative energy bound state.
\end{remark}
Now we are ready to prove Theorem $\ref{linear hair theorem}$.
\begin{proof}
By Lemma $\ref{charged negative energy bound state}$, Lemma $\ref{nonempty}$, and Proposition $\ref{eigenvalue}$, we only need to find an $M$ such that $\lambda_{M} = 0$. First, we derive the monotonicity of $\lambda_{M}$ and $\lambda_{M}^{\epsilon}$ in terms of $M$. For any $f\in C_{c}^{\infty}(r_{+},\infty)$ and $M_{1}>M_{2}$, we have\begin{align*}
L_{M_{1}}^{\epsilon}[f] =& \int_{r_{+}+\epsilon}^{\infty}r^{2}\Omega_{M_{1}}^{2}\left(\frac{df}{dr}\right)^{2}+\left(\frac{1}{l^{2}}\alpha-\frac{q_{0}^{2}e_{1}^{2}\left(\frac{1}{r_{+}}-\frac{1}{r}\right)^{2}}{\Omega_{M_{1}}^{2}}\right)r^{2}f^{2}\mathrm{d}r\\
 =& \int_{r_{+}+\epsilon}^{\infty}r^{2}(\Omega_{M_{1}}^{2}-\Omega_{M_{2}}^{2})\left(\frac{df}{dr}\right)^{2}+\left(\frac{q_{0}^{2}e_{2}^{2}\left(\frac{1}{r_{+}}-\frac{1}{r}\right)^{2}}{\Omega_{M_{2}}^{2}}-\frac{q_{0}^{2}e_{1}^{2}\left(\frac{1}{r_{+}}-\frac{1}{r}\right)^{2}}{\Omega_{M_{1}}^{2}}\right)r^{2}f^{2}\mathrm{d}r\\+&L_{M_{2}}^{\epsilon}[f]
\end{align*}
Since for fixed $(r_{+},l)$ and $r$, $\Omega_{M}^{2}$ is a decreasing function of $M$ and $\frac{e^{2}}{\Omega_{M}^{2}}$ is an increasing funciton of $M$, we have \begin{equation*}
L_{M_{1}}^{\epsilon}[f]\leq L_{M_{2}}^{\epsilon}[f].
\end{equation*}
Passing to the limit, we have $\lambda_{M}$ and $\lambda_{M}^{\epsilon}$ are increasing functions of $M$. By the same computation as above, we can further show that $\lambda_{M}^{\epsilon}$ is a continuous function of $M$:\begin{equation*}
0\leq\lambda_{M_{1}}^{\epsilon}-\lambda_{M_{2}}^{\epsilon}\leq C(M_{1}-M_{2}),
\end{equation*} 
where $C$ is a uniform constant independent of $\epsilon$. Passing to the limit, we get that $\lambda_{M}$ is also an increasing continuous function of $M$. Let $M_{c} = \inf_{M}\mathcal{A}_{s}$, where $s = M_{b}$ if one considers the large charge case and $s = M_{0}$ if one considers the general fixed charge case. Then we can continuously extend $\lambda_{M}$ to $M = M_{c}$ by letting\begin{equation*}
\lambda_{M_{c}} = \lim_{M\rightarrow M_{c}}\lambda_{M}\geq 0.
\end{equation*}
If $\lambda_{M_{c}}>0$, then by remark $\ref{remark}$ and the construction above, we can find $\phi_{c}^{\epsilon}\in\mathcal{F}^{\epsilon}$ and $\lambda_{M_{c}}^{\epsilon}>0$ such that $\phi^{\epsilon}_{c}$ is the solution of\begin{equation*}
-\frac{d}{dr}(r^{2}\Omega_{M_{c}}^{2}\frac{d\phi_{c}^{\epsilon}}{dr})+\left(\frac{1}{l^{2}}\alpha-\frac{q_{0}^{2}A_{t}^{2}}{\Omega_{M_{c}}^{2}}\right)r^{2}\phi_{c}^{\epsilon} = -\lambda_{M_{c}}^{\epsilon}\frac{r^{2}}{\Omega_{M_{c}}^{2}}\phi_{c}^{\epsilon}.
\end{equation*}
Then $L_{M_{c}}^{\epsilon}[\phi_{c}^{\epsilon}]<0$ implies we can find a negative energy bound state for $L_{M_{c}}$, which means $M_{c}\in\mathcal{A}_{s}$. Since $\lambda_{M}$ is a contninuous function of $M$, for $M_{c}-\epsilon<M<M_{c}$ where $\epsilon$ is a small positive number, we have $\lambda_{M}<0$. Then we can apply the above argument in Proposition $\ref{eigenvalue}$ again to show that $M\in\mathcal{A}_{s}$, which contradicts to the definition of $M_{c}$. Therefore $\lambda_{M_{c}} = 0$. 

Last, we need to construct the corresponding function $\phi_{M_{c}}$. We consider the solution $\phi_{M,\lambda}$ of the equation\begin{equation*}
-\frac{d}{dr}(r^{2}\Omega_{M}^{2}\frac{d\phi}{dr})+\left(\frac{1}{l^{2}}\alpha-\frac{q_{0}^{2}A_{t}^{2}}{\Omega_{M}^{2}}\right)r^{2}\phi+\lambda\frac{r^{2}}{\Omega_{M}^{2}}\phi = 0
\end{equation*}
with $\phi(r)\approx(r-r_{+})^{\frac{1}{\kappa}\sqrt{\lambda}}$ and Dirichlet boundary condition. By the asymptotic analysis, we write $\phi$ as \begin{equation*}
\phi(r,M,\lambda) = B_{1}(M,\lambda)\phi_{D}(r,M,\lambda)+B_{2}(M,\lambda)\phi_{N}(r,M,\lambda),
\end{equation*} 
where \begin{equation*}
\phi_{D}(r,M,\lambda)\approx r^{-\frac{3}{2}-\Delta},\quad \phi_{N}(r,M,\lambda)\approx r^{-\frac{3}{2}+\Delta},\quad r\rightarrow\infty.
\end{equation*}
Since $B_{2}(M,\lambda_{M}) = 0$ for $M\in\mathcal{A}_{s}$, we have $B_{2}(M_{c},\lambda_{M_{c}}) = 0$ by continuity. Then \begin{equation*}
\phi_{M_{c}} := \phi(r,M_{c},0)
\end{equation*}
is the non-zero solution of $\ref{stationary Klein-Gordon}$ with Dirichlet boundary condition and can be extended continuously to the event horizon.
\end{proof}

\section{Proof of Theorem $\ref{linear hair theorem}$ for Neumann boundary conditions}
\label{Neumann section}
In this section, we begin to prove Theorem $\ref{linear hair theorem}$ for Neumann boundary conditions. We will elaborate on the new ingredients in this different boundary condition case while omitting some proofs similar to those used in the Dirichlet boundary condition case.
\subsection{Outline of the proof}
\label{outline of neumann}
In Section $\ref{outline of the proof}$, we discussed the outline of the proof for Dirichlet boundary conditions. Based on the discussion there and the proofs we used in Section $\ref{proof of linear hair theorem}$, we further discuss the new challenges we will face in the case of Neumann boundary conditions.

First, if we want to apply a similar variational method, the immediate difficulty we will face is that the energy functional $L_{M}[f]$ is not even finite for functions with Neumann boundary conditions since functions with Neumann boundary conditions decay more slowly near infinity. This can be overcome by using the appropriately designed twisted derivative, which was first introduced by Breitenlohner and Freedman \cite{breitenlohner1982stability} and later was used in \cite{warnick2013massive,dold2017unstable,holzegel2014boundedness,holzegel2014boundedness} to study the Klein--Gordon equations on asymptotic AdS spacetimes under Neumann boundary conditions.

Second, to achieve Neumann boundary conditions for the minimizer, we can no longer take $f$ in Lemma \ref{change} to be compactly supported. Then Lemma \ref{change} will fail in general since the boundary term generated from using the integration by parts in the proof of Lemma $\ref{change}$ is not finite. Recall that in the proof of Proposition \ref{nontrival solution}, we need to use a different but in principle equivalent twisted energy functional to show that the limit $\phi$ is non-trivial. This step has not worked for the Neumann boundary conditions since the failure of Lemma \ref{change} for generic twist functions. Hence, we have to come up with a more robust method. We overcome this difficulty by constructing a nice twisted energy functional.

Third, as in the case of the Dirichlet boundary conditions, we can only expect local convergence $\phi^{\epsilon}\rightarrow\phi$ while boundary conditions restrict the behavior of $\phi$ at $r = \infty$. For the Dirichlet boundary conditions, in the proof of Proposition $\ref{nontrival solution}$, $L_{M}[\phi]$ is finite if and only if the solution $\phi$ satisfies the Dirichlet boundary condition. We can obtain the finiteness of the energy functional on any compact set $K$ and then pass to the limit to get the finiteness of $L_{M}[\phi]$. This strategy does not work for the Neumann boundary conditions, since the suitable twisted energy functional $L_{M,h}[\phi]$ is finite for functions with Neumann or Dirichlet boundary conditions and we lose information at infinity when applying the local convergence argument. The boundary condition of $\phi$ is achieved by establishing the uniform bound of $\phi^{\epsilon}$, analogously to Proposition \ref{gain regularity}
\subsection{More general twisted derivatives}
In this section, we introduce a new function space and the twisted derivatives we will use.

First, to illustrate the idea of dealing with non-integrability of the energy function generated by the slow decay nature of the function with Neumann boundary conditions, we consider the following naive twisted energy functional.

Let $h = r^{-\frac{3}{2}+\Delta}$. We can rewrite the equation $\eqref{stationary Klein-Gordon}$ by using the twisted derivative $h\partial_{r}(h^{-1}\cdot)$:
\begin{align}
&-r^{\frac{3}{2}-\Delta}\frac{d}{dr}\left(r^{2\Delta-1}\Omega_{M}^{2}\frac{dr^{\frac{3}{2}-\Delta}\phi}{dr}\right)+\left(V_{M,\frac{3}{2}-\Delta}-\frac{q_{0}^{2}A_{t}^{2}}{\Omega_{M}^{2}}\right)r^{2}\phi = 0,\\
&V_{M,\frac{3}{2}-\Delta} = \frac{(\frac{3}{2}-\Delta)(\Delta-\frac{1}{2})}{r^{2}}+\frac{2(\frac{3}{2}-\Delta)^{2}M}{r^{3}}-\frac{(\frac{3}{2}-\Delta)(\frac{5}{2}-\Delta)e^{2}}{r^{4}}.
\end{align}
The twisted energy functional for this equation is \begin{equation}
L_{M,\frac{3}{2}-\Delta}[f] = \int_{r_{+}}^{\infty}r^{2\Delta-1}\Omega_{M}^{2}\left(\frac{dr^{\frac{3}{2}-\Delta}f}{dr}\right)^{2}+\left(V_{M,\frac{3}{2}-\Delta}-\frac{q_{0}^{2}A_{t}^{2}}{\Omega_{M}^{2}}\right)r^{2}f^{2}dr,
\end{equation}
which is finite for functions with Neumann boundary conditions. The secret in the above choice of the twisted derivative is that the leading order term in the potential is canceled. However, in this situation, the leading term of $V_{M,\frac{3}{2}-\Delta}$ is comparable to the term $\frac{q_{0}A_{t}^{2}}{\Omega_{M}^{2}}$, which means the positivity of the potential term near $r=\infty$ depends on the values of $(M,r_{+},l,\alpha,q_{0})$. In view of Remark \ref{coercivity of infinity}, the method in the Dirichlet boundary condition case fails for Neumann boundary conditions. To make our method more robust and cover all ranges of possible $\alpha$ where the negative energy bound state and local well-posedness hold, we define the following twist function $h_{0}$, which can be viewed as a suitable modification of the twist function $r^{-\frac{3}{2}+\Delta}$ preserving its structure at infinity:\begin{equation}
\label{twisted function}
h_{0} = \begin{cases}
r^{-\frac{3}{2}+\Delta}(1-r^{-1-\Delta}),\quad &r\geq r_{+}+2,\\
g,\quad &r_{+}+1<r<r_{+}+2,\\
e^{-Kr},\quad &r_{+}\leq r<r_{+}+1,
\end{cases}
\end{equation}
where $K$ is a large constant which will be chosen later and $g$ is a smooth function that makes $h$ smooth on $(r_{+},\infty)$. When $r\rightarrow\infty$, we have
\begin{equation}
\label{positive near infinity}
\begin{aligned}
&\lim_{r\rightarrow\infty}r^{1+\Delta}\left(V_{M,h}(r)-\frac{q_{0}^{2}A_{t}^{2}}{\Omega_{M}^{2}}\right)\\=& \lim_{r\rightarrow\infty}r^{1+\Delta}\left(-\frac{1+\Delta}{r^{\frac{1}{2}+\Delta}(1-r^{-1-\Delta})}\frac{d}{dr}\left(r^{-\frac{3}{2}}\Omega_{M}^{2}\right)+\frac{(\frac{3}{2}-\Delta)(1+\Delta)}{r^{\frac{1}{2}+\Delta}(1-r^{-1-\Delta})}r^{-\frac{5}{2}}\Omega_{M}^{2}\right)\\=&\frac{1}{l^{2}}
\left(-\alpha-\frac{5}{4}\right)>0.
\end{aligned}
\end{equation}
When $r\rightarrow r_{+}$, we have\begin{equation}
\label{positive near horizon}
\lim_{r\rightarrow r_{+}}V_{M,h_{0}}(r)-\frac{q_{0}^{2}A_{t}^{2}}{\Omega_{M}^{2}}=\frac{1}{l^{2}}\alpha+K\kappa>0
\end{equation}
for $K$ large enough.

To define the new function space, let $\chi$ be a smooth function which is $1$ on $[r_{+},R)$ and vanishes near infinity. Then we can define $\mathcal{F}_{N}$:
\begin{equation*}
\mathcal{F}_{N}: = \left\{f\in L^{2}(r_{+},\infty),\quad \chi f\in C_{c}^{\infty}(r_{+},\infty),\quad \int_{r_{+}}^{\infty}\frac{r^{2}}{\Omega_{M}^{2}}f^{2}\mathrm{d}r = 1\right\},
\end{equation*}
where $f\in\mathcal{F}_{N}$ can be viewed as a function supported away from $r = r_{+}$. The difference between $\mathcal{F}_{N}$ and $\mathcal{F}$ defined in Section \ref{proof of linear hair theorem} is that $f$ is no longer compactly supported. Consequently, Lemma \ref{change} is not true generically since the non-vanishing boundary terms.

\subsection{Negative energy bound state and the eigenvalue solution}
\label{negative energy sec}
Since to find the negative energy bound state, we only need to consider $L_{M,h_{0}}[f]$ with $f$ compactly supported. Hence, we can immediately get the existence of the negative energy bound state for $L_{M,h_{0}}[f]$. Specifically, we have the following two lemmas, the proofs of which follow line by line from the case of Dirichlet boundary conditions in Lemma \ref{charged negative energy bound state} and Lemma \ref{nonempty}.
\begin{lemma}
\label{charged negative energy bound state Neumann}
For any fixed parameters $(r_{+},l,\alpha)$ satisfying the bound $-\frac{9}{4}<\alpha<-\frac{5}{4}$, and for any $M_{e = 0}<M_{b}<M_{0}$, there exists a $q_{1}(M_{b},r_{+},l,\alpha)>0$ such that for any $\vert q_{0}\vert>q_{1}$, $\mathcal{A}_{M_{b}}$ is non-empty.
\end{lemma}
\begin{lemma}
\label{nonempty Neumann}
For every fixed parameters $(r_{+},l,\alpha,q_{0})$ satisfying\begin{align}
&-\frac{9}{4}<\alpha<\left\{-\frac{3}{2}+\frac{1}{2}l^{2}q_{0}^{2},-\frac{5}{4}\right\},\\&
\frac{r_{+}^{2}}{l^{2}}>R_{0},
\end{align}
where $R_{0}$ is defined in \eqref{quadratic}, $\mathcal{A}_{M_{0}}$ is non-empty.
\end{lemma}
Next, we prove the existence of a minimizer for the energy functional $L_{M,h_{0}}^{\epsilon}[f]$. 
\begin{proposition}
If $L_{M,h_{0}}[f]$ has a negative energy bound state, then for any $\epsilon>0$ small enough, $L_{M,h_{0}}^{\epsilon}[f]$ can attain its negative minimum $L_{M,h_{0}}^{\epsilon}[\phi^{\epsilon}]$in the function class $\mathcal{F}_{N}^{\epsilon}$. Moreover, $\phi^{\epsilon}$ satisfies the Neumann boundary condition.
\label{eigenvalue for Neumann}
\end{proposition}
\begin{proof}
The proof here is almost the same as the proof of Proposition \ref{eigenvalue}. Consider the twisted energy functional\begin{equation*}
L_{M,h_{0}}^{\epsilon}[f] = \int_{r_{+}+\epsilon}^{\infty}r^{2}\Omega_{M}^{2}h_{0}^{2}\left(\frac{dh_{0}^{-1}f}{dr}\right)^{2}+\left(V_{M,h_{0}}(r)-\frac{q_{0}^{2}A_{t}^{2}}{\Omega_{M}^{2}}\right)r^{2}f^{2}dr.
\end{equation*}
By the limit $\eqref{positive near infinity}$, we have $V_{M,h_{0}}-\frac{q_{0}^{2}A_{t}^{2}}{\Omega_{M}^{2}}$ is positive and asymptotically $r^{-1-\Delta}$ when $r\rightarrow\infty$. By the limit $\eqref{positive near horizon}$, we have $V_{M,h_{0}}-\frac{q_{0}^{2}A_{t}^{2}}{\Omega_{M}^{2}}$ is positive near the event horizon. Let $x_{1}$ and $x_{2}$ be the smallest and largest root of $V_{M,h_{0}}-\frac{q_{0}^{2}A_{t}^{2}}{\Omega_{M}^{2}}$ respectively. We have
\begin{equation}
\label{coercivity}
\begin{aligned}
L_{M,h_{0}}^{\epsilon}[f] =& \int_{r_{+}+\epsilon}^{\infty}r^{2}\Omega_{M}^{2}h_{0}^{2}\left(\frac{dh_{0}^{-1}f}{dr}\right)^{2}dr+\left(\int_{r_{+}+\epsilon}^{x_{1}}+\int_{x_{2}}^{\infty}\right)\left(V_{M,h_{0}}-\frac{q_{0}^{2}A_{t}^{2}}{\Omega_{M}^{2}}\right)r^{2}f^{2}dr\\&-
\int_{x_{1}}^{x_{2}}\left(\frac{q_{0}^{2}A_{t}^{2}}{\Omega_{M}^{2}}-V_{M,h_{0}}\right)r^{2}f^{2}dr\\\geq&
-C\int_{x_{1}}^{x_{2}}r^{-1-\Delta}r^{2}f^{2}dr\geq -C\int_{x_{1}}^{x_{2}}f^{2}dr>-C.
\end{aligned}
\end{equation}
Hence $L_{M,h_{0}}^{\epsilon}$ is lower bounded in $\mathcal{F}^{\epsilon}_{N}$. Let $\phi_{n}^{\epsilon}$ be the minimizing sequence. Then we have\begin{equation}
\left\Vert\frac{d\phi_{n}^{\epsilon}}{dr}\right\Vert_{L^{2}}\leq C(\epsilon).
\end{equation}
Therefore $\phi_{n}^{\epsilon}$ is $H^{1}$ bounded. Then we have $\phi_{n}^{\epsilon}$ is weakly convergent to $\phi^{\epsilon}$ and strongly convergent to $\phi^{\epsilon}$ on any compact set $K\subset (r_{+},\infty)$. Then by $\eqref{coercivity}$, we have in fact\begin{equation*}
\int_{r_{+}+\epsilon}^{\infty}r^{1-\Delta}(\phi_{n}^{\epsilon})^{2}dr<C.
\end{equation*}
Hence by the same argument as in Proposition $\ref{eigenvalue}$, we have $$\int_{r_{+}+\epsilon}^{\infty}\frac{r^{2}}{\Omega_{M}^{2}}(\phi^{\epsilon})^{2}\mathrm{d}r = 1$$.

Since $\phi^{\epsilon}$ is the minimizer of the energy functional $L_{M,h_{0}}^{\epsilon}[f]$, by variational principle (see Chapter $8$ in \cite{evans2022partial}), we have $\phi^{\epsilon}$ solves the equation
\begin{equation}
\label{phi epsilon equation}
-h_{0}^{-1}\frac{d}{dr}(r^{2}\Omega_{M}^{2}h_{0}^{2}\frac{dh^{-1}_{0}\phi^{\epsilon}}{dr})+\left(V_{M,h_{0}}(r)-\frac{q_{0}^{2}A_{t}^{2}}{\Omega_{M}^{2}}\right)r^{2}\phi^{\epsilon} = -\lambda_{M}^{\epsilon}\frac{r^{2}}{\Omega_{M}^{2}}\phi^{\epsilon}
\end{equation}
with Neumann boundary condition. Last, note we can also write the equation \eqref{phi epsilon equation} in the untwisted form:
\begin{equation}
-\frac{d}{dr}\left(r^{2}\Omega_{M}^{2}\frac{d\phi^{\epsilon}}{dr}\right)+\left(\frac{1}{l^{2}}\alpha-\frac{q_{0}^{2}A_{t}^{2}}{\Omega_{M}^{2}}\right)r^{2}\phi^{\epsilon} = -\lambda_{M}^{\epsilon}\frac{r^{2}}{\Omega_{M}^{2}}\phi^{\epsilon}.
\end{equation}
\end{proof}
\begin{remark}
One should note that the Dirichlet boundary condition in the proof of Proposition \ref{eigenvalue} is obtained by looking at the finite energy. However, since functions in our function class $\mathcal{F}_{N}$ no longer vanish near infinity, the variational principle itself gives information about the boundary condition of $\phi^{\epsilon}$.
\end{remark}
We can get rid of $\epsilon$ in the following proposition:\begin{proposition}
\label{non-trival solution Neumann}
If there exists a negative energy bound state of $L_{M,h_{0}}[f]$, then we can find a non-zero solution $\phi_{M}$ of the equation\begin{equation}
\label{eq for phiM}
-\frac{d}{dr}\left(r^{2}\Omega_{M}^{2}\frac{d\phi_{M}}{dr}\right)+\left(\frac{1}{l^{2}}\alpha-\frac{q_{0}^{2}A_{t}^{2}}{\Omega_{M}^{2}}\right)r^{2}\phi_{M} = -\lambda_{M}\frac{r^{2}}{\Omega_{M}^{2}}\phi_{M}.
\end{equation}
Moreover, $\phi_{M}$ satisfies the Neumann boundary condition and can be extended continuously to the event horizon $r = r_{+}$.
\end{proposition}
\begin{proof}
By using the same argument as in the proof of Proposition $\ref{nontrival solution}$, we can show that $\phi^{\epsilon}$ is weakly convergent to $\phi_{M}$ in $H^{1}_{loc}$ and strongly convergent to $\phi_{M}$ in $L^{2}_{loc}$. By $\eqref{coercivity}$, we have\begin{equation*}
\int_{x_{1}}^{x_{2}}(\phi^{\epsilon})^{2}dr>C.
\end{equation*}
Hence $\phi_{M}$ is non-trival and solves the equation\begin{equation}
-h_{0}^{-1}\frac{d}{dr}\left(r^{2}\Omega_{M}^{2}h_{0}^{2}\frac{dh_{0}^{-1}\phi}{dr}\right)+\left(V_{M,h_{0}}-\frac{q_{0}^{2}A_{t}^{2}}{\Omega_{M}^{2}}\right)r^{2}\phi = -\lambda_{M}\frac{r^{2}}{\Omega_{M}^{2}}\phi,
\end{equation}
which is equivalent to \eqref{eq for phiM}.

It remains to prove that $\phi_{M}$ still satisfies the Neumann boundary condition. This part does not follow trivially since $\phi_{M}$ is obtained by the local convergence and we may lose information near infinity.
For any compact set $K\subset(r_{+},\infty)$, by using the same argument as in Proposition $\ref{gain regularity}$, we have\begin{equation}
\left\vert\int_{K}r^{2\Delta+1}\frac{d}{dr}\left(r^{\frac{3}{2}-\Delta}\phi^{\epsilon}\right)r^{\frac{-1-\Delta}{2}}\mathrm{d}r\right\vert<C
\end{equation}
is uniformly bounded in $\epsilon$ and $K$. Then, by the weak convergence results, we have 
\begin{equation}
\label{contraction}
\left\vert\int_{K}r^{2\Delta+1}\frac{d}{dr}\left(r^{\frac{3}{2}-\Delta}\phi_{M}\right)r^{\frac{-1-\Delta}{2}}\mathrm{d}r\right\vert<C.
\end{equation}
However, if $\phi_{M}$ has the Dirichlet branch, then $\eqref{contraction}$ can not be uniformly bounded when $K$ is located near infinity, which is a contradiction.
\end{proof}
As in the Dirichlet boundary condition case, one immediate corollary of Proposition \ref{non-trival solution Neumann} is that one can construct growing mode solutions to the uncharged Klein--Gordon equation \eqref{Klein-Gordon} for all $M_{c}<M<M_{0}$.

Now, by the exact same argument as in the proof of Theorem $\ref{linear hair theorem}$ for Dirichlet boundary conditions, we can finish the proof of Theorem $\ref{linear hair theorem}$ for Neumann boundary conditions.

\section{Proof of growing mode solution}
In this section, we finally close the proof of Theorem $\ref{growing mode theorem}$.
\subsection{Growing mode solutions for the large charge case and the general fixed charge case under Dirichlet boundary condition}
\label{growing mode solution}
\begin{theorem}
Let $(M_{c},r_{+},l,\alpha,q_{0})$ be the parameters where the stationary solutions $\phi_{0}$ defined in Theorem \ref{linear hair theorem} exist. Then there exists $\epsilon_{0}>0$ such that there exist analytic functions $M_{e = 0}\leq M(\epsilon)<M_{0}$ and $\omega_{R}(\epsilon)\in\mathbb{R}$ for $-\epsilon_{0}<\epsilon\leq\epsilon_{0}$ with $M(0) = M_{c}$ and $\omega_{R}(0) = 0$ such that mode solutions $\phi_{\epsilon} = e^{i(\omega_{R}+i\epsilon)}\psi_{\epsilon}$ to the equation \eqref{Klein-Gordon} with parameters $(M(\epsilon),r_{+},l,\alpha,q_{0})$ exist under the Dirichlet boundary condition. Moreover, $\phi$ can be continuously extended to the event horizon $\{r = r_{+}\}$ and \begin{align}
&\frac{dM}{d\epsilon}(0)<0,\label{diff1}\\&
q_{0}e\frac{d\omega_{R}}{d\epsilon}(0)<0,\quad q_{0}e\neq 0\label{diff2}.
\end{align}
\end{theorem}
We can write the solution of $\eqref{Klein-Gordon equation for psi}$ in the form of \begin{equation}
\psi = \widetilde{A}(M,\omega)u_{D}(r,M,\omega)+B(M,\omega)u_{N}(r,M,\omega),\label{near infinity of psi}
\end{equation}
where $\{u_{D},u_{N}\}$ is the local basis of solutions of $\eqref{Klein-Gordon equation for psi}$ with $u_{D}$ satisfying the Dirichlet boundary condition and $u_{N}$ satisfying the Neumann boundary condition. Note that we use the notation $\widetilde{A}$ here to distinguish it from the electromagnetic potential $A$. Furthermore, if $\omega$ is a real number while $u_{D}$ and $u_{N}$ are not real functions, then $\overline{u_{D}}$ and $\overline{u_{N}}$ are also two linearly independent solutions and satisfy the same boundary conditions. Hence, we can always take $u_{D}$ and $u_{N}$ to be real if $\omega$ is real.
\begin{proof}
Since $\phi_{0}$ is a stationary solution, we have $B(M_{c},0) = 0$. Let $\omega = \omega_{R}+i\omega_{I}$, $\widetilde{A} = A_{R}+iA_{I}$, and $B = B_{R}+iB_{I}$, then by the implicit function theorem, to prove the existence of growing mode solutions under Dirichlet boundary conditions, we only need to show \begin{equation}
\label{det condition}
\det\begin{bmatrix}
\frac{\partial B_{R}}{\partial \omega_{R}} & \frac{\partial B_{R}}{\partial M}\\
\frac{\partial B_{I}}{\partial{\omega_{R}}} & \frac{\partial B_{I}}{\partial M}
\end{bmatrix}
\neq 0.
\end{equation}
For $\omega = \omega_{R}$ real, multiplying the equation \eqref{stationary Klein-Gordon} by $\overline{\psi}_{0}$ and taking the imaginary part, we have\begin{equation}
\frac{d}{dr}\Im\left(r^{2}\Omega_{M}^{2}\frac{d\psi}{dr}\bar{\psi}\right) = 0
\end{equation}

Then we can derive
\begin{equation}
\Im\left(r^{2}\Omega_{M}^{2}\frac{d\psi}{dr}\overline{\psi}\right)(r_{+}) = \lim_{r\rightarrow\infty}\Im\left(r^{2}\Omega_{M}^{2}\frac{d\psi}{dr}\overline{\psi}\right)(r).
\end{equation}
By the near horizon and near infinity behaviors of $\phi_{0} = \psi_{0}$ in $\eqref{near horizon for psi}$ and $\eqref{near infinity of psi}$, we have\begin{align*}
r_{+}^{2}\omega_{R} = \lim_{r\rightarrow\infty}\frac{1}{l^{2}}\left(\Im(A\overline{B})u_{D}\frac{\partial u_{N}}{\partial r}+\Im(\overline{A}B)u_{N}\frac{\partial u_{D}}{\partial r}\right) = \frac{2}{l^{2}}\Delta(A_{I}B_{R}-A_{R}B_{I}).
\end{align*}
Taking the $\omega_{R}$ and $M$ derivatives respectively and evaluating at $(\omega_{R}=0, M = M_{c})$, we have\begin{equation}
\begin{aligned}
A_{I}\frac{\partial B_{R}}{\partial\omega_{R}}-A_{R}\frac{\partial B_{I}}{\partial\omega_{R}}& = \frac{l^{2}r_{+}^{2}}{2\Delta} ,\\
A_{I}\frac{\partial B_{R}}{\partial M}-A_{R}\frac{\partial B_{I}}{\partial M} & = 0.
\end{aligned}
\end{equation}
Since $(A_{R},A_{I})(M_{c},0)\neq 0$, $\eqref{det condition}$ holds if and only if $\frac{\partial B}{\partial M}\neq 0 $. Assume \begin{equation}
\label{contradiction argument}
\frac{\partial B}{\partial M} = 0.
\end{equation}
Then differentiating the equation $\eqref{Klein-Gordon equation for psi}$ at $(M = M_{c},\omega = 0)$, we have\begin{equation}
\label{equation for psi M}
\frac{d}{dr}\left(r^{2}\Omega_{M}^{2}\frac{d}{dr}\frac{\partial\psi_{0}}{\partial M}\right)+\frac{d}{dr}\left(r^{2}\frac{\partial \Omega_{M}^{2}}{\partial M}\frac{d\psi_{0}}{dr}\right) = \left(\frac{1}{l^{2}}\alpha-\frac{q_{0}^{2}A_{t}^{2}}{\Omega_{M}^{2}}\right)r^{2}\frac{\partial \psi_{0}}{\partial M}-\frac{\partial}{\partial M}\left(\frac{q_{0}^{2}A_{t}^{2}}{\Omega_{M}^{2}}\right)r^{2}\psi_{0}.
\end{equation}
We can multiply the above equation $\eqref{equation for psi M}$ by $\overline{\psi}_{0}$ and use the integration by parts. Since by our assumption $\eqref{contradiction argument}$, the boundary terms that appear in the integration by parts vanish. Hence we have\begin{equation}
\label{dM equations}
\begin{aligned}
&\int_{r_{+}}^{\infty}\left(\frac{d}{dr}\left(r^{2}\Omega_{M}^{2}\frac{d\overline{\psi}_{0}}{dr}\right)-\left(\frac{1}{l^{2}}\alpha-\frac{q_{0}^{2}A_{t}^{2}}{\Omega_{M}^{2}}\right)r^{2}\overline{\psi}_{0}\right)\frac{\partial\psi_{0}}{\partial M}dr \\=& \int_{r_{+}}^{\infty}r^{2}\frac{\partial\Omega_{M}^{2}}{\partial M}\left\vert\frac{d\psi_{0}}{dr}\right\vert^{2}-\frac{\partial}{\partial M}\left(\frac{q_{0}^{2}A_{t}^{2}}{\Omega_{M}^{2}}\right)r^{2}\vert\psi\vert^{2}dr.
\end{aligned}
\end{equation}
Since $\overline{\psi}_{0}$ is also the solution of $\eqref{Klein-Gordon equation for psi}$, the left-hand side of the above equation is $0$. Since $\frac{\partial\Omega^{2}}{\partial M}<0$ on $(r_{+},\infty)$ for parameters $(M,r_{+},l)$, we have\begin{equation*}
\psi_{0} = \frac{d\psi_{0}}{dr} = 0,
\end{equation*}
which contradicts the fact that $\psi_{0}$ is non-trivial. Hence $\frac{\partial B}{\partial M}\neq 0$ and by the implicit function theorem, for $\epsilon$ small enough, there exist real analytic functions $\omega_{R}(\epsilon)$ and $M(\epsilon)$ on $(-\epsilon,\epsilon)$, such that for parameters $\left(M(\epsilon),r_{+},l,\alpha,q_{0}\right)$, we have the growing mode solution $\phi(r,\epsilon) = e^{i\omega_{R}(\epsilon)t}e^{-\epsilon t}\psi(r,\epsilon)$ satisfying the Dirichlet boundary condition.

To prove \eqref{diff1} and \eqref{diff2}, differentiating the equation \eqref{Klein-Gordon equation for psi} with respect to $\epsilon$ at $\epsilon = 0$, we have\begin{equation}
\label{equ:deri of ep}
-\frac{d}{dr}\left(r^{2}\frac{\partial}{\partial\epsilon}\bigl(\Omega_{M}^{2}\frac{d\psi_{0}}{dr}\bigr)\right)+\left(\frac{1}{l^{2}}\alpha-\frac{q_{0}^{2}A_{t}^{2}}{\Omega_{M}^{2}}\right)r^{2}\frac{\partial\psi_{0}}{\partial\epsilon} = 
\left(q_{0}^{2}\frac{\partial}{\partial\epsilon}\bigl(\frac{A_{t}^{2}}{\Omega_{M}^{2}}\bigr)-\frac{2q_{0}A_{t}}{\Omega_{M}^{2}}\bigl(\frac{\partial\omega_{R}}{\partial\epsilon}+i\bigr)\right)r^{2}\psi_{0}.
\end{equation}
By the boundary conditions \eqref{decay of growing mode near horizon} and \eqref{near horizon for psi} for $\psi_{0}$ at the event horizon, we have\begin{align}
&\frac{\partial}{\partial\epsilon}\bigl(\Omega_{M}^{2}\frac{d\psi_{0}}{dr}\bigr)\overline{\psi_{0}}(r_{+}) = -1+i\frac{d\omega_{R}}{d\epsilon},\\&
\lim_{r\rightarrow r_{+}}\Omega_{M}^{2}\frac{d\overline{\psi}_{0}}{dr}\frac{\partial\psi_{0}}{\partial\epsilon}(r,0) = 0
\end{align}
Since $\psi(r,\epsilon)$ satisfies the Dirichlet boundary condition, the boundary terms at infinity vanish. For the equation \eqref{equ:deri of ep}, multiplying $\overline{\psi}_{0}$, taking the real and imaginary part respectively, and integrating by parts, we have\begin{align}
&r_{+}^{2}\frac{d\omega_{R}}{d\epsilon} = \int_{r_{+}}^{\infty}\frac{-2q_{0}A_{t}}{\Omega_{M}^{2}}r^{2}\vert\psi_{0}(r,0)\vert^{2}\mathrm{d}r,\\&
r_{+}^{2} = \int_{r_{+}}^{\infty}r^{2}\frac{\partial\Omega_{M}^{2}}{\partial\epsilon}\left\vert\frac{d\psi_{0}}{dr}\right\vert^{2}+\left(-q_{0}^{2}\frac{\partial}{\partial\epsilon}\bigl(\frac{A_{t}^{2}}{\Omega_{M}^{2}}\bigr)+\frac{2q_{0}A_{t}\frac{\partial\omega_{R}}{\partial\epsilon}}{\Omega_{M}^{2}}\right)r^{2}\vert\psi_{0}\vert^{2}\mathrm{d}r.
\end{align}
Hence we have\begin{align*}
&q_{0}e\frac{d\omega_{R}}{d\epsilon}<0,\\&
\frac{dM}{d\epsilon}<0.
\end{align*}
This concludes the proof.
\end{proof}
\subsection{Growing mode solutions for the weakly charged case under Dirichlet boundary condition}
The existence of growing mode solutions to the uncharged Klein--Gordon equation under Dirichlet boundary conditions, as stated in the weakly charged case of Theorem \ref{growing mode theorem}, follows from Proposition \ref{nontrival solution}, Proposition \ref{imaginary omega pro}, and the equation \eqref{Klein-Gordon equation for psi}. It remains to show the existence of growing mode solutions near the extremality when $q_{0}\neq0$ is small. Let $M_{c}$ and $\lambda_{M}$ for $M_{c}<M<M_{0}$ as in Theorem \ref{linear hair theorem}, we have the following result:
\begin{theorem}
Let $\phi = e^{\sqrt{\lambda_{M}}t}\psi$ be the growing mode solution to the uncharged Klein--Gordon equation \eqref{Klein-Gordon} with parameters $(M,r_{+},l,\alpha,q_{0} = 0)$ under Dirichlet boundary. Then there exist a small positive number $\epsilon_{0}>0$ and real analytic functions $\omega_{R}(\epsilon)$ and $\omega_{I}(\epsilon)<0$ on $(-\epsilon_{0},\epsilon_{0})$ with $\omega_{I}(0) = -\sqrt{\lambda_{M}}$ such that a growing mode solution of the form $e^{i(\omega_{R}(\epsilon)+i\omega_{I}(\epsilon))}\psi_{\epsilon}$ to \eqref{Klein-Gordon} with parameters $(M,r_{+},l,\alpha,\epsilon)$ under the Dirichlet boundary condition exists.
\end{theorem}
\begin{proof}
We can write the solution of $\eqref{Klein-Gordon equation for psi}$ in the form of \begin{equation}
\psi = \widetilde{A}(q_{0},\omega_{R},\omega_{I})u_{D}(r,q_{0},\omega_{R},\omega_{I})+B(q_{0},\omega_{R},\omega_{I})u_{N}(r,q_{0},\omega_{R},\omega_{I}),\label{near infinity of growing psi}
\end{equation}
as in Section \ref{growing mode solution}. Assume that for parameters $(M,r_{+},l,\alpha,q_{0} = 0)$, the uncharged Klein--Gordon equation \eqref{Klein-Gordon} admits a growing mode solution $e^{i\omega t}\psi$ with mode $\omega = -i\sqrt{\lambda_{M}}$ under the Dirichlet boundary condition. The boundary conditions for $\psi$ at the event horizon are given by $\psi(r_{+}) = 0$ and $\Omega_{M}^{2}\frac{d\psi}{dr} = i\omega\psi$. We can deduce that\begin{equation}
\frac{d}{d\omega_{I}}(\Omega_{M}^{2}\frac{d\psi}{dr})\bar{\psi} = \frac{d}{d\omega_{R}}(\Omega_{M}^{2}\frac{d\psi}{dr})\bar{\psi} = 0.
\end{equation}

Differentiating the equation \eqref{Klein-Gordon equation for psi} with $\omega_{R}$ and $\omega_{I}$ respectively at $q_{0} = 0$ and $\omega_{I} = -\sqrt{\lambda_{M}}$, taking the imaginary part, and integrating over $(r_{+},\infty)$, analogously to the computation in \eqref{dM equations}, we have\begin{align}
&A_{R}\frac{\partial B_{I}}{\partial \omega_{R}}-A_{I}\frac{\partial B_{R}}{\partial \omega_{R}} = \frac{l^{2}\sqrt{\lambda_{M}}}{\Delta},\\&
A_{R}\frac{\partial B_{I}}{\partial\omega_{I}}-A_{I}\frac{\partial B_{R}}{\partial\omega_{I}} = 0.
\end{align}
Analogously to the proof of Theorem \ref{growing mode solution}, assuming $\frac{\partial B}{\partial\omega_{I}} = 0$, differentiating equation \eqref{Klein-Gordon equation for psi} with respect to $\omega_{I}$, multiplying $\overline{\psi}$, and integrating by parts over $(r_{+},\infty)$, we can get a contradiction. Hence $\frac{\partial B}{\partial\omega_{I}}\neq0$. Then, by the implicit function theorem, we can finish the proof.
\end{proof}

\subsection{Growing mode solution for Neumann boundary condition}
The proof of Theorem $\ref{growing mode theorem}$ under Neumann boundary conditions follows almost line by line as in the previous section by considering the twisted equation \eqref{twisted Klein-Gordon equation for psi} for $\psi$.

\paragraph{Conflict of Interest Statement} The author declares that there is no conflict of interest.
\paragraph{Data Availability Statement} The author declares that there is no associated data.

\bibliographystyle{plain}
\bibliography{main_checked}
\end{document}